\documentclass[lettersize,journal]{IEEEtran}

\normalsize

\ifCLASSINFOpdf
\else
\fi

\usepackage[T1]{fontenc}
\usepackage{mathrsfs}
\usepackage{bm}
\usepackage{soul}
\usepackage{amsmath}
\usepackage{breqn}
\usepackage{amsthm}
\usepackage{amssymb}
\usepackage{graphicx}
\PassOptionsToPackage{normalem}{ulem}
 \usepackage{subfig}
\usepackage{setspace}
\usepackage{cite}
\usepackage{hyperref}
\usepackage{xcolor}
\usepackage[normalem]{ulem}
\useunder{\uline}{\ul}{}
\makeatletter

\theoremstyle{plain}
\newtheorem{theorem}{\protect\theoremname}
 \theoremstyle{plain}
 \newtheorem{conjecture}{\protect\conjecturename}
 \theoremstyle{plain}
 
 \theoremstyle{plain}
  \newtheorem{lemma}{\protect\lemmaname}
 \theoremstyle{remark}
 \newtheorem{remark}{\protect\remarkname}
  \newtheorem{assumption}{\protect\assumptionname}
\theoremstyle{assumption}
  \theoremstyle{proposition}
 \newtheorem{proposition}{\protect\propositionname}
 
\theoremstyle{algorithm}

\setkeys{Gin}{width=1.0\columnwidth}

\makeatother

\usepackage{babel}
 \providecommand{\definitionname}{Definition}
 \providecommand{\lemmaname}{Lemma}
 \providecommand{\propositionname}{Proposition}
 \providecommand{\remarkname}{Remark}
\providecommand{\theoremname}{Theorem}
\providecommand{\conjecturename}{Conjecture}
\providecommand{\assumptionname}{Assumption}
\providecommand{\algorithmname}{Algorithm}

\begin{document}

 \title{ Age-of-information minimization via opportunistic sampling by an energy harvesting source
 \thanks{Akanksha Jaiswal is with the Department of Electrical Engineering, Indian Institute of Technology, Delhi. Email:  akanksha.jaiswal@ee.iitd.ac.in.  Arpan Chattopadhyay is with the Department of Electrical Engineering and the Bharti School of Telecom Technology and Management, Indian Institute of Technology, Delhi. Email:  arpanc@ee.iitd.ac.in.   Amokh Varma is with the Department of Mathematics, Indian Institute of Technology, Delhi. Email: mt6180527@maths.iitd.ac.in. }
 \thanks{A. C. acknowledges support via the professional development fund and professional development allowance from IIT Delhi, grant no. GP/2021/ISSC/022 from I-Hub Foundation for Cobotics and grant no. CRG/2022/003707 from Science and Engineering Research Board (SERB), India.}

\thanks{The preliminary version of this paper appeared in \cite{jaiswal2021minimization}.}}

 \author{
Akanksha Jaiswal, Arpan Chattopadhyay and Amokh Varma\\

 }
 
\maketitle
\begin{abstract} 
Herein, minimization of time-averaged age-of-information (AoI) in an energy harvesting (EH) source setting is considered. The EH source opportunistically samples one or multiple processes over discrete time instants and sends the status updates to a sink node over a wireless fading channel. Each time, the EH node decides whether to probe the link quality and then decides whether to sample a process and communicate based on the channel probe outcome. The trade-off is between the freshness of information available at the sink node and the available energy at the source node. We use infinite horizon Markov decision process (MDP) to formulate the AoI minimization problem for two scenarios where energy arrival and channel fading processes are: (i) independent and identically distributed (i.i.d.), (ii) Markovian. In i.i.d. setting, after channel probing, the optimal source sampling policy is shown to be a threshold policy. Also, for unknown channel state and EH characteristics, a variant of the Q-learning algorithm is proposed for the two-stage action model, that seeks to learn the optimal policy.  For Markovian system, the problem is again formulated as an MDP, and a learning algorithm is provided for unknown dynamics. Finally, numerical results demonstrate the policy structures and performance trade-offs.
\end{abstract}

\begin{IEEEkeywords}
Age-of-information, remote sensing, energy harvesting,  Markov decision process (MDP), Reinforcement learning.
\end{IEEEkeywords}

\section{Introduction}\label{section:introduction}
In recent years, the need for combining physical systems with the cyber world has attracted significant research interest. These cyber-physical systems (CPS) are supported by ultra-low power, low latency Internet of Things (IoT) networks, and encompass a large number of applications such as vehicle tracking, environment monitoring, intelligent transportation, industrial process monitoring, smart home systems, etc. Such systems often require the deployment of sensor nodes to monitor a physical process and send real-time status updates to a remote estimator over a wireless network. However, for such critical CPS applications, minimizing mean packet delay without accounting for delay jitter can often be detrimental to the system's performance. Also, mean delay minimization does not guarantee delivery of the observation packets to the sink node in the same order in which they were generated, thereby often resulting in unnecessarily dedicating network resources towards delivering outdated observation packets despite the availability of a freshly generated observation packet in the system. Hence, it is necessary to take into account the freshness of information of the data packets, apart from the mean packet delay. 

Recently, a metric named Age of Information (AoI) has been proposed \cite{kaul2012real} as a measure of the freshness of the information update. In this setting, a sensor monitoring system generates a time-stamped status update and sends it to the sink over a network. At the time $t$, if the latest monitoring information available to the sink node comes from a packet whose time-stamped generation instant was $t'$, then the AoI at the sink node is computed as $(t-t’)$. Thus, AoI has emerged as an alternative metric to mean delays \cite{talak2018can}.

 However, the timely delivery of status updates is often limited by energy and bandwidth constraints in the network. Recent efforts towards designing EH source nodes ({\em e.g.}, source nodes equipped with solar panels) have opened a new research paradigm for IoT network operations. The energy generation process in such nodes is very uncertain and typically modeled as a stochastic process. The harvested energy is stored in an energy buffer as energy packets and used for sensing and communication as and when needed. This EH capability significantly improves network lifetime and eliminates the need for frequent manual battery replacement, but poses a new challenge towards network operations due to uncertainty in the available energy at the source nodes at any given time. 
 
 Motivated by the above challenges, we consider the problem of minimizing the time-averaged expected AoI in a remote sensing setting, where a single EH source probes the channel state, samples one or multiple processes, and sends the observation packets to the sink node over a fading channel. The energy generation process is modeled as a discrete-time i.i.d. process, and a finite energy buffer is considered. Two variants of the problem are considered: (i) single process with channel state information at the transmitter (CSIT), and (ii) multiple processes with CSIT. Channel state probing and process sampling for time-averaged expected AoI minimization is formulated as an MDP, and the threshold nature of the optimal policy is established analytically for each case. Next, we propose reinforcement learning (RL) algorithms to find the optimal policy when the channel statistics and energy harvesting characteristics are unknown. We also extend these works to the setting where the system has a Markovian channel and Markovian energy harvesting process.  In addition, we consider AoI minimization without channel probing, and compare its performance numerically with our algorithms that involve channel probing. Interestingly,     channel probing yields smaller AoI if the energy required in probing is significantly smaller than sampling energy.  Numerical results validate the theoretical results and intuitions.

 \subsection{Related work} 
 Initial efforts towards optimizing AoI mostly involved the analysis of various queueing models; {\em e.g.,} \cite{kaul2012real} for analyzing a single source single server queueing system with FCFS service discipline, \cite{kaul2012status} for LCFS service discipline for M/M/1 queue, \cite{yates2018age} and \cite{kaul2020timely} for multi-source single sink system with M/M/1 queueing at each source, \cite{kosta2019age} for AoI performance analysis for multi-source single-sink infinite-buffer queueing system where unserved packets are substituted by available newer ones, etc. Additionally, the paper \cite{fountoulakis2023scheduling} has developed a scheduling policy for two users- a deadline-constrained user, and an AoI-oriented user that send their information to a single receiver through an error-prone channel (modeled as both i.i.d.  and time-correlated channel). The authors formulated the average AoI minimization problem with timely throughput constraints as a constrained Markov decision process (CMDP). The authors in \cite{chen2023minimizing} have considered the AoI-minimization problem for the peak power-constrained base station (BS) which sends updates to multiple users through time-varying, Markovian channels, and proposed a truncated multi-user scheduling policy. The optimal threshold-based sampling policy for AoI minimization for a source sending samples to a remote estimator over an unreliable channel with feedback has been proposed in \cite{pan2022optimal}. 
 
 On the other hand, a number of papers have considered  AoI minimization problem under EH setting: \cite{farazi2018age} for the derivation of average AoI for a single source having finite battery capacity, \cite{arafa2017age} for the derivation of the minimal age policy for EH two-hop network, \cite{farazi2018average} for average AoI expression for single source EH server, \cite{hu2018age} for AoI minimization for a wirelessly powered user, \cite{arafa2017age1} for sampling, transmission scheduling and transmit power selection for a single source single sink system over infinite time horizon where the delay is dependent on the packet transmission energy. The authors in 
 \cite{chen2019age} have considered two source nodes (power grid node and EH sensor node) sending different data packets to a common destination by using multiple access channel, and derived the delay and AoI of two source nodes respectively. The authors of \cite{zheng2019closed} have analyzed the non-linear AoI for EH system, whereas \cite{abd2022age, abd2022closed } focused on the characterization of moment generating function of AoI in EH status update systems under various queueing disciplines.  In \cite{khorsandmanesh2021average}, average AoI and peak AoI (PAoI) are analyzed for a source and a relay network where energy is harvested from a power station.

 There have also been several other works on developing optimal scheduling policy for minimizing AoI for EH sensor networks \cite{wu2017optimal,yang2015optimal,bacinoglu2018achieving,feng2021age,bacinoglu2015age,leng2019age,gindullina2021age,ceran2021learning,arafa2018age,tunc2019optimal,bacinoglu2017scheduling,arafa2021timely}. For example, \cite{wu2017optimal} has investigated optimal online policy for single sensor single sink system with a noiseless channel for infinite, finite and unit size battery capacity; for finite battery size, it has provided energy aware status update policy. The paper 
 \cite{yang2015optimal} has considered a multi-sensor single sink system with {\em infinite} battery size, and proposed a randomized myopic scheduling scheme. For an EH source with finite battery and Poisson energy arrival process, the authors in \cite{bacinoglu2018achieving} have provided age-energy trade-off and optimal threshold policy for AoI minimization, but unlike our work, they have not considered probing/estimation of a fading channel while making a sampling decision. 
 In \cite{feng2021age}, the optimal online status update policies to minimize the long-run average AoI for EH source with updating erasures have been proposed. It has been shown that the best-effort uniform updating policy is optimal when there is no feedback and the best-effort uniform updating with re-transmission (BUR) policy is optimal when feedback is available to the source. The authors of 
 \cite{bacinoglu2015age} have examined the problem of minimizing AoI under a constraint on the count of status updates. The authors of \cite{leng2019age} have addressed the AoI minimization problem for cognitive radio communication systems with EH capability; they formulated optimal sensing and updating for perfect and imperfect sensing as a partially observable Markov decision process (POMDP). Information source diversity, {\em i.e.,} multiple sources tracking the same process but with dissimilar energy cost, and sending status updates to the EH monitoring node with finite battery size, has been considered in \cite{gindullina2021age} with an MDP formulation, but no structure was provided for the optimal policy. In \cite{ceran2021learning}, reinforcement learning has been used to minimize AoI for a single EH sensor with HARQ protocol where a decision is made whether to send new status updates or retransmit failed status updates and it has shown the existence of a threshold-based optimal policy. In \cite{tunc2019optimal}, the authors have formulated optimal transmission schemes for an EH sensor that sends update packets to the BS by considering different transmission modes (controlled by power and error probability) and battery recovery setting as an MDP. The authors of \cite{bacinoglu2017scheduling} have developed a threshold policy for minimizing AoI for a single-sensor single-sink system with an erasure channel and no channel feedback. For a system with Poisson energy arrival, unit battery size, and error-free channel, it has shown that a threshold policy achieves an average age lower than that of a zero-wait policy; based on this, a lower bound on average age for general battery size and erasure channel has been derived.
 
 The paper \cite{abd2020reinforcement}  has proposed a threshold-based sampling policy for AoI minimization in a multi-source system and also uses deep reinforcement learning to learn the optimal policy. Here, it is assumed that the source nodes harvest energy from the destination, but it cannot simultaneously harvest energy and transmit data by using downlink and uplink fading channels, respectively. The authors in \cite{hatami2021aoi}  have proposed optimal threshold policy and RL algorithms for a system with cache-enabled edge node between users and EH sensors, which either commands the sensor to send an update or retrieves the age measurement from the cache on user request. However, these papers \cite{ceran2021learning,abd2020reinforcement,hatami2021aoi} have not considered channel probing as done in our paper, which makes our system model completely different from them and it involves challenges to find out optimal probing and sampling decisions taken by the source node (capable of harvesting energy and transmit data simultaneously) at each time instant. Also, we have formulated the AoI minimization problem and proposed RL algorithms for both the i.i.d. system as well as the Markovian system which involves a two-stage action model (probing and sampling), which is not available in the literature.
  
  Usually, a source node probes the channel by sending a pilot signal which consumes resources. Hence, the source node needs to identify the optimal channel probing and data transmission strategies, in order to optimize energy consumption.  Several papers have considered channel probing policies for wireless communication systems \cite{ chang2007optimal,hentati2019energy,chaporkar2008optimal,johnston2013channel}. The authors in \cite{yu2023aoi} have considered a system model where an energy-limited robot sends short status packets to a control center (CC) over a wireless fading channel which is estimated by sending a pilot signal during each coherence time. They have formulated the average AoI minimization problem as a CMDP with an average power consumption constraint and proposed a threshold-based status update scheme. However, these papers have not used channel probing policies for AoI minimization under EH remote sensing system, which motivates us to include channel probing before transmitting information updates, and the major challenges brought by the consideration of probing lies in the formulation of two-stage MDP model which involves two actions, i.e., channel probing and sampling decision-making at the same time. The value iteration used here for finding the optimal solution involves intermediate iterates which is non-standard. Also, this setup has not been considered before. Further, for unknown dynamics, Q-learning formulation involving a two-stage action model is not straightforward and different from that available in the literature.  Later, in Fig. \ref{fig5}, we demonstrate that the channel probing is useful in minimizing AoI under practical system parameters settings.

 \vspace{-2.5mm}

\subsection{ Contributions}
\begin{enumerate}
  \item We formulate the problem of minimizing the time-averaged expected AoI in an EH remote sensing system with a single source monitoring one or multiple processes, as an MDP with a two-stage action model. The two stages of actions involve channel probing and process sampling at each time instant. The considered MDP with a two-stage action model is different from standard MDP in the literature since it involves using certain {\em intermediate value function}. Under the assumptions of i.i.d. time-varying channel with CSIT, channel state probing capability at the source, and finite battery size, for the single/multiple process(es), we derive the optimal sampling policy structures which turn out to be a simple threshold policy on the probed channel quality which is a function of the age of the process(es). The source node, depending on the current age of a process (the largest age of the processes), decides whether to probe the channel or not. Next, based on the channel probe outcome, the source node decides whether to sample the process (with the largest age) and send an observation packet, or to remain idle. Thus, the MDP involves taking action in two stages at each time instant where a sampling decision is taken only if the channel is probed.
  
   \item We prove the convergence of an analogue of value iteration for this two-stage MDP.  Our analysis directly shows that the difference between the value function iterate and the optimal value function decreases to $0$ exponentially fast.
  
  \item We prove certain interesting properties of various cost functions which are useful for deriving optimal threshold structures. For example, we prove that the primary cost value functions increase with age, and the intermediate cost-to-go value functions decrease with packet success probability in the probed channel. 
    \item We also formulate the AoI minimization problem for a Markovian system model with a source monitoring a single process as an MDP; the channel dynamics and energy harvesting characteristics are modeled as Markov chains in this case. We provide optimality equations for the MDP. Certain conjectures regarding the threshold structure of the optimal policy for this MDP are supported numerically. Note that,  these results can be extended to multiple processes scenarios, but we have omitted this case due to lack of space.  
  \item  For a source sensing a single process with unknown channel statistics and energy harvesting characteristics, we propose reinforcement learning algorithms that yield the optimal policies for the formulated MDPs. However, unlike standard Q-learning algorithm involving a single action at each time instant, we propose a variant of the Q-learning algorithm, that involves taking the action in two stages; first, the source decides whether to probe the channel quality and if the channel quality turns out to be good, then it further decides whether to sample and communicate based on its available energy and age in sampling. Similar techniques can also be used for multiple processes.  The Q-learning algorithms are based on asynchronous stochastic approximation techniques \cite{borkar1998asynchronous}.
\end{enumerate}
\subsection{ Organization}
The rest of the paper is organized as follows. The system model has been explained in Section~\ref{section:system-model}. AoI minimization for both single and multiple process case with i.i.d. system model is addressed in Section~\ref{section:single-sensor-single-process}, whereas AoI minimization for Markovian system model is considered in Section~\ref{section:Markovian channel and Markovian energy}. Further, reinforcement learning for the i.i.d. and the Markovian system model is proposed in Section~\ref{section:RL}. Lastly, numerical results are provided in Section~\ref{section:numerical-work}, followed by the conclusions in Section~\ref{section:conclusion}. All proofs are provided in the appendices.  

{\bf Notation:} Throughout this paper, $\mathbb{P}(\cdot)$
 and $\mathbb{E}(\cdot)$ denote the probability and expectation operators, respectively. For ease of reference, the other key notations used in the paper are listed in Table I.

 \vspace{-8pt}
\section{System model}\label{section:system-model} 
We consider an EH source capable of sensing one out of $N$ different processes at a time, and reporting the observation packet to a sink node over a fading channel; see Fig.~\ref{system-model}. We assume that the source node can sense the processes perfectly without any error and loss.
 Time is discretized with the discrete-time index $t \in \{0, 1,2,3, \cdots\}$. At each time, the source node can decide whether to estimate the quality of the channel from the source to the sink or not. If the source node decides to probe the channel state, it can further decide whether to sample a process and communicate the data packet to the sink or not, depending on the instantaneous channel quality. The source has a finite energy buffer of size $B$ units, where $E_{p}$ unit of buffer energy is used to probe the channel once and $E_{s}$ unit of buffer energy is used in sensing and communication of one packet.  We assume that the source can simultaneously harvest energy and transmit the data packet using previously stored energy in the buffer; see \cite{wang2015design, das2021effect, ulukus2015energy} for reference to solar panel assisted wireless communication nodes where energy harvesting and packet transmission happen simultaneously.  

\begin{figure}[h]
  \begin{center}
 \includegraphics[height=2.5cm,width=7cm]{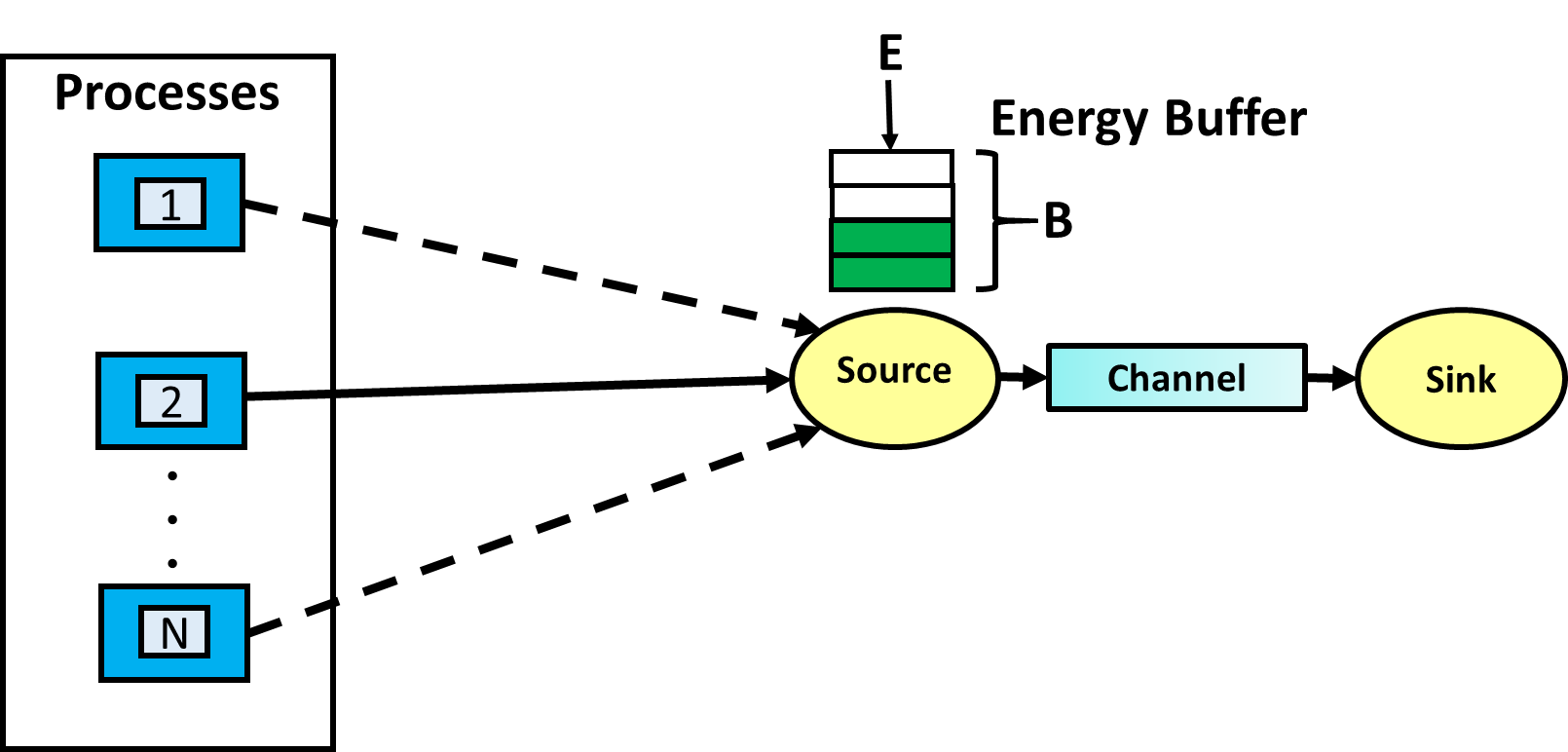}
 \caption{Pictorial representation of a remote sensing system where an EH source samples one of $N$ number of processes at a time and sends the observation packet to a sink node.}
 \label{system-model}
 \end{center}
 \vspace{-5mm}
\end{figure}

 \vspace{-8pt}
\subsection{Channel model}\label{subsection:system_Channel model}
We consider a fading channel between the source and the sink, where a particular fading state is characterized by the probability of success of a transmitted packet. 
We denote by $p(t) \in \{p_1, p_2, \cdots, p_m\}$ the probability of packet transmission success from the source to the sink node at time~$t$, and by $C(t) \in \{C_1, C_2,\cdots, C_m\}$ the corresponding channel state, where $m$ is the finite number of channel states. The packet success probability corresponding to channel state $C_j$ is given by $p(C_j)=p_j$. Let us also denote by $r(t) \in \{0,1\}$ the indicator that the packet transmission from the source to the sink at time~$t$ is successful. Obviously,  $\mathbb{P}(r(t)=1|C(t)=C_j)=p_j$ for all $j \in \{1,2,\cdots,m\}$. It is assumed that the channel state $C(t)$ is learnt perfectly via a channel probe. For many practical systems, the data packets that need to be transmitted are very long and the pilot signal used for channel estimation is short in duration and requires very less time as compared to the time required for data transmission. Therefore, the channel can be estimated instantaneously without a penalty in time.

\noindent {\bf I.I.D. channel model: }
In this case, we assume that $\{p(t)\}_{t \geq 0}$ is i.i.d. across $t$, with $\mathbb{P}(p(t)=p_j)=q_j$ for all $j \in \{1,2,\cdots,m\}$  \cite{fountoulakis2023scheduling}.  

\noindent{\bf Markovian channel model: } Here, we model the fading channel as finite state Markov chain \cite{wang1995finite, zhang1999finite, yao2021age, fountoulakis2023scheduling}, where the one-step state transition probability from state $C_i$ to state $C_j$ is denoted by $q_{i,j}$ and the $t$-step transition probability is given by $q_{i,j}^{(t)} = \mathbb{P}(C(t) = C_j|C(0) =C_i )$ for all $i, j \in \{1,2,\cdots,m\}$.  

 \vspace{-8pt}
\subsection{Energy harvesting process}\label{subsection:system_Energy harvesting process}
Let $A(t)$ denote the number of energy packet arrivals to the energy buffer at time~$t$, and $E(t)$ denote the energy available to the source at time~$t$, for all $t \geq 0$. 
We consider the  following two  models for energy harvesting:

\noindent{\bf I.I.D. model: } The energy packet generation process $\{A(t)\}_{t \geq 0}$ in the energy buffer is assumed to be an i.i.d. process with known mean $\lambda>0$;  see \cite{hatami2021aoi, gindullina2021age, leng2019age, chen2019age} for reference.

\noindent{\bf Markov model: }  To address the temporal correlation in energy generated by many renewable energy sources such as wind and solar, the papers  \cite{michelusi2013transmission,sombabu2020age} have considered a two-state Markov chain to model the harvested energy. Motivated by this, we model the EH state of the source as a two-state Markov chain with states $\{H_{1}, H_{2}\}$, where $H_{1}$ is called the harvesting state and $H_{2}$ is called the non-harvesting state. When the source node is in the harvesting state, the energy packet generation process is considered to be a (conditionally) i.i.d. process (see \cite{sombabu2020age}) with known mean $\lambda>0$. When the source node is in a non-harvesting state, then the energy packet generation rate is zero. The sate transition probability from state $H_{1}$ to state $H_{2}$ is $h_{1,2}$ and from state $H_{2}$ to state $H_{1}$ is $h_{2,1}$ and it is governed by a Markov chain. Thus, the energy generation process is a Markov-modulated i.i.d. process and it exhibits temporal correlation.  

We assume that, in case the energy buffer at the source node is full, the newly generated energy packets will not be accommodated unless $E_{p}$ unit of energy is spent in probing.
 \vspace{-8pt}
\subsection{Decision process and policy}
At time~$t$, let $b(t)\in \{0,1\}$ denote the indicator of deciding to probe the channel, and $a(t) \in \{0,1,\cdots,N\}$ denote the identity of the process being sampled, with $a(t)=0$ meaning that no process is sampled, and $b(t)=0$ meaning that the channel is not probed. Also, $b(t)=0$ implies $a(t)=0$. The set of possible actions or decisions is denoted by $\mathcal{A}=\{\{0,0\}\cup \{1\times \{0,1,2,\cdots,N\}\}\}$, where a generic action at time~$t$ is denoted by $(b(t),a(t))$.

Let us denote by $\tau_k(t)\doteq \sup\{0 \leq \tau < t: a(\tau)=k, r(\tau)=1\}$ the last time instant before time~$t$, when process~$k$ was sampled and the observation packet was successfully delivered to the sink. The age of information (AoI) for the $k$-th process at time~$t$ is  given by $T_k(t)=(t-\tau_k(t))$. However, if $a(t)=k$ and $r(t)=1$, then $T_k(t)=0$ since the current observation of the $k$-th process is available to the sink node. 
A generic Markov scheduling policy is a collection of mappings $\{\mu_t\}_{t \geq 0}$, where $\mu_t$ maps   $(E(t), C(t), \{T_k(t)\}_{1 \leq k \leq N})$ to the action space $\mathcal{A}$. If $\mu_t=\mu$ for all $t \geq 0$, the policy is called stationary, else non-stationary. 

We seek to find a stationary scheduling policy $\mu$ that minimizes the expected sum AoI which is averaged over time:
\begin{align}\label{eqn:main-problem}
\footnotesize
 \min_{\mu} \frac{1}{T } \sum_{t=0}^T \sum_{k=1}^N \mathbb{E}_{\mu} (T_k(t))
 \normalsize
\end{align}


\begin{table}[t!]
\caption{Key Notations Used}
\label{tab:my-table}
\resizebox{\columnwidth}{!}{%
\begin{tabular}{|l|l|}
\hline
Notation Used & Description                                                                                     \\ \hline
$t$           & Discrete time index                                                                             \\ \hline
$B$           & Energy buffer size                                                                              \\ \hline
$E_p$         & Energy required for probing the channel                                                         \\ \hline
$E_s$         & Energy required for sensing and communicating a data packet                                     \\ \hline
$N$           & Number of processes                                                                             \\ \hline
$C_j$         & Channel state $j$, $\forall$ $j$ $\in$ $\{1, 2, \cdots, m\}$, where $m$ denotes a finite number \\ \hline
$p_j$         & Packet success probability of channel state $C_j$                                               \\ \hline
$q_j$         & Probability of being in a channel state $C_j$                                                   \\ \hline
$\lambda$     & Energy arrival rate                                                                             \\ \hline
$A$           & Number of energy packet arrivals  in a slot                                                              \\ \hline
$E$           & Energy available at the energy buffer                                                           \\ \hline
$T_k$         & AoI of the $k$-th process                                                                         \\ \hline
$H_1$         & Harvesting state of the Markovian channel model                                                 \\ \hline
$H_2$         & Non-harvesting state of the Markovian channel model                                             \\ \hline
$h_{1,2}$     & The transition probability from state $H_1$ to state $H_2$                                      \\ \hline
$r$             & Indicator of successful packet transmission                                                     \\ \hline
$b$             & Indicator of probing channel state                                                              \\ \hline
$a$             & Indicator of sampling a process                                                                 \\ \hline
\end{tabular}%
}
\end{table}

 \vspace{6pt}
\section{I.I.D. system} \label{section:single-sensor-single-process}
In this section, we derive the optimal channel probing, source activation, and data transmission policy for a single EH source sampling a single process ($N=1$) or multiple processes, when the channel and energy harvesting processes follow the i.i.d. model as described in Section~\ref{section:system-model}.

 \vspace{-8pt}
\subsection{Single process ($N=1$)}
 Here, we formulate (1) as a long-run average cost MDP with state space $\mathcal{S} \doteq \{0,1,\cdots,B\} \times \mathbb{Z}_+$ and an intermediate state space $\mathcal{V}=\{0, 1,…..,B\}\times \mathbb{Z}_{+} \times \{C_{1}, C_{2},\cdots, C_{m}\}$.  A generic primary state $s = (E,T)$   means that the energy buffer has $E$ energy packets, and the last successfully received packet was generated at the source node $T$ slots ago.  A generic intermediate state $v = (E,T,C)$  additionally means that the current channel state obtained via probing is  $C$. The action space is $\mathcal{A}=\{\{0,0\}\cup \{1\times\{0,1 \}\}\}$ with $a(t), b(t) \in \{0,1\}$. At each time, if the source node decides not to probe the channel state, then it will not perform sampling, and thus $b(t) = 0, a(t) = 0$, and the expected single-stage AoI cost is $c(s(t), b(t), a(t)) = T$. However, if the source node decides to probe the channel state, the expected single-stage AoI cost is $c(s(t), b(t)=1, a(t)=0) = T$, and $c(s(t), b(t)=1, a(t)=1) = T(1-p(C))$, where the expectation is taken over packet success probability $p(C)$. We first formulate the average-cost MDP problem as an $\alpha$-discounted cost MDP problem with $\alpha \in(0, 1)$, and derive the optimal policy, from which the solution of the average cost minimization problem can be obtained by taking $\alpha \rightarrow 1$  \cite[Section 4.1]{bertsekas2011dynamic}.

\subsubsection{Optimality equation}
Let $J^{*}(E,T)$ be the optimal value function for primary state $(E,T)$ in the discounted cost problem, and let $W^*(E,T,C)$ be the cost-to-go from an intermediate state $(E,T,C)$. The Bellman equations \cite[Proposition 7.2.1]{bertsekas2005dynamic}  for the  $\alpha$-discounted cost MDP problem are given by \eqref{eqn:Bellman-eqn-single-sensor-single-process-with-fading-general} which are solved by standard value iteration technique \cite[Proposition 7.3.1]{bertsekas2005dynamic}.

\scriptsize
\begin{eqnarray} \label{eqn:Bellman-eqn-single-sensor-single-process-with-fading-general}
J^{*}(E \geq E_{p}+E_{s},T)&=&min \bigg\{T+\alpha \mathbb{E}_{A}J^{*}(min\{E+A,B\},T+1), \nonumber\\
&&V^{*}(E,T) \bigg\}\nonumber\\
V^{*}(E,T)&=& \sum_{j=1}^m q_{j} W^{*}(E,T,C_{j})\nonumber\\ 
W^{*}(E,T,C)&=&min\{T+\alpha \mathbb{E}_{A}J^{*}(min\{E-E_{p}+A,B\},T+\nonumber\\ 
&&1),T(1-p(C))+\alpha p(C)\mathbb{E}_{A}J^{*}(min\{E-E_{p}\nonumber\\
&&-E_{s}+A,B\},1)+\alpha (1-p(C))\mathbb{E}_{A}J^{*}(min\{E\nonumber\\
&&-E_{p}-E_{s}+A,B\},T+1)\}\nonumber\\
J^{*}(E< E_{p}+E_{s},T)&=&T+\alpha \mathbb{E}_{A}J^{*}(min\{E+A,B\},T+1)
\end{eqnarray}
\normalsize

The first expression in the minimization in the R.H.S. of the first equation in \eqref{eqn:Bellman-eqn-single-sensor-single-process-with-fading-general} is the cost of not probing channel state ($b(t)=0$), which includes single-stage AoI cost $T$ and an $\alpha$ discounted future cost for a random next state $(min\{E+A,B\},T+1)$, averaged over the distribution of the number of energy packet generation $A$. The quantity $V^{*}(E,T)$ is the expected cost of probing the channel state, which explains the second equation in \eqref{eqn:Bellman-eqn-single-sensor-single-process-with-fading-general}. At an intermediate state $(E,T,C)$, if $a(t)=0$, a single stage AoI cost $T$ is incurred and the next state becomes $(min\{E-E_p+A,B\},T+1)$; if $a(t)=1$, the expected AoI cost is $T(1-p(C))$ (expectation taken over the packet success probability $p(C)$), and the next random state becomes $(min\{E-E_p-E_s+A,B\},1)$ and $(min\{E-E_p-E_s+A,B\},T+1)$ if $r(t)=1$ and $r(t)=0$, respectively. The last equation in \eqref{eqn:Bellman-eqn-single-sensor-single-process-with-fading-general} follows similarly since $b(t)=0, a(t)=0$ is the only possible action when $E<E_p+E_s$. 

Substituting the value of $V^{*}(E,T)$ in the first equation of \eqref{eqn:Bellman-eqn-single-sensor-single-process-with-fading-general}, we obtain the following Bellman equations:

 \scriptsize
\begin{eqnarray}\label{eqn:Bellman-eqn-single-sensor-single-process-with-fading}
&&J^{*}(E \geq E_{p}+E_{s},T)\nonumber\\
&=&min \bigg\{T+\alpha \mathbb{E}_{A}J^{*}(min\{E+A,B\},T+1),\mathbb{E}_{C} \bigg( min\{T+\nonumber\\ 
&&\alpha \mathbb{E}_{A}J^{*}(min\{E-E_{p}+A,B\},T+1),T(1-p(C))+\nonumber\\ 
&&\alpha p(C)\mathbb{E}_{A}J^{*}(min\{E-E_{p}-E_{s}+A,B\},1)+\nonumber\\
&&\alpha (1-p(C))\mathbb{E}_{A}J^{*}(min\{E-E_{p}-E_{s}+A,B\},T+1)\} \bigg) \bigg\} \nonumber\\
&&J^{*}(E < E_{p}+E_{s},T)\nonumber\\
&=&T+\alpha \mathbb{E}_{A}J^{*}(min\{E+A,B\},T+1) 
\end{eqnarray}
\normalsize

\subsubsection{Policy structure}\label{subsubsection:Policy structure-IID-System-N=1}
We first provide the convergence proof of value iteration since we have a two-stage decision process as opposed to traditional MDP where a single action is taken.
\begin{proposition} \label{proposition:convergence-of-value-function-J}
The value function $J^{(t)}(s)$ converges to $J^{*}(s)$ as   $t \rightarrow \infty$.
\end{proposition}

\begin{proof}See Appendix ~\ref{appendix:proof-of-convergence-of-value-function-J}.
\end{proof}
Next, we provide some properties of the value function, which will help in proving the structure of the optimal policy.
\begin{lemma} \label{lemma:single-sensor-single-process-J-decreasing-in-p}
For $N=1$,  
$J^{*}(E, T)$ is increasing in $T$ and $W^{*}(E,T,C)$ is decreasing in $p(C)$. 
\end{lemma}

\begin{proof} See Appendix ~\ref{appendix:proof-of-lemma-single-sensor-single-process-fading-cost-increasing-in-T-decrease-in-p}.
\end{proof}
\begin{conjecture}\label{conjecture:single-sensor-single-process-with-fading-policy-structure}
 For $N=1$, the optimal probing policy for the $\alpha$-discounted AoI cost minimization problem is a threshold policy on $T$. For any $E \geq E_{p}+E_{s} $, the optimal action is to probe the channel state if and only if $T\geq T_{th}(E)$ for a threshold function $T_{th}(E)$ of $E$. 
\end{conjecture}

\begin{remark}
    The optimal probing decision at $(E,T)$ is obtained by using the condition that the expected cost of probing is less than or equal to the expected cost of not probing. This condition, when further simplified, reduces to checking whether $f(E,T) \geq 0$ for a suitable function $f$. From this, one can prove Conjecture~$1$  if $f(E,T)$ can be shown to be increasing in $T$, which is difficult in this case. A similar difficulty is also faced in other subsequent system models.
\end{remark}

\begin{theorem}\label{theorem:single-sensor-single-process-with-fading-policy-structure}
 For $N=1$, at any time, if the source decides to probe the channel, then the optimal sampling policy is a threshold policy on $p(C)$. For any $E \geq E_{p}+E_{s} $ and probed channel state, the optimal action is to sample the source node if and only if $p(C)\geq p_{th}(E,T)$ for a threshold function $p_{th}(E,T)$ of $E$ and $T$. 
\end{theorem}

\begin{proof} See Appendix ~\ref{appendix:proof-of-theorem-single-sensor-single-process-fading-threshold-policy}.
\end{proof}

The policy structure supports the following two intuitions. Firstly, given $E$ and $T$, the source decides to probe the channel state if AoI is greater than some threshold value. 
Secondly, given $E$ and $T$ and probed channel state, if the channel quality is better than a threshold, then the optimal action is to sample the process and communicate the observation to the sink node. We will later numerically observe in Section~\ref{section:numerical-work} some intuitive properties of $T_{th}(E)$ as a function of $E$ and $\lambda$, and $p_{th}(E,T)$ as a function of $E$, $T$ and $\lambda$.\\

\subsubsection{No probing case}\label{No probing case}
If we  consider the case where the source   samples the process directly without channel probing, then the Bellman equations for this case are given by:

\scriptsize
\begin{eqnarray} \label{eqn:Bellman-eqn-single-sensor-single-process-without-probing}
J^{*}(E \geq E_{s},T)&=&min \bigg\{T+\alpha \mathbb{E}_{A}J^{*}(min\{E+A,B\},T+1), \nonumber\\
&& \mathbb{E}_{C}\bigg( T(1-p(C))+\alpha p(C)\mathbb{E}_{A}J^{*}(min\{E-E_{s}+\nonumber\\
&&A,B\},1)+\alpha (1-p(C))\mathbb{E}_{A}J^{*}(min\{E-E_{s}+A,\nonumber\\
&&B\},T+1) \bigg) \bigg\}\nonumber\\
J^{*}(E< E_{s},T)&=&T+\alpha \mathbb{E}_{A}J^{*}(min\{E+A,B\},T+1) 
\end{eqnarray}
\normalsize

 Here also, we numerically observe an optimal threshold policy on $T$; see Section~\ref{section:numerical-work}.

 \vspace{-12pt}
\subsection{Multiple processes ($N>1$)}\label{subsection:single-sensor-multiple-process}

Here, we formulate the $\alpha$-discounted cost version of \eqref{eqn:main-problem} as an MDP with a generic primary state $s=(E,T_1,T_2,\cdots, T_N)$ which means that the energy buffer has $E$ energy packets, and the latest (successfully) received packet from the  $k$-th  process was generated at the source $T_{k}$ slots ago.  Also, a generic intermediate state is $v=(E,T_1,T_2,\cdots, T_N,C)$ which additionally means that the current channel state $C$ is learnt by probing, has packet success probability $p(C)$. The action space $\mathcal{A}=\{\{0,0\}\cup \{1\times\{0,1,2,\cdots,N \}\}\}$ with $b(t)\in \{0,1\}$ and $a(t) \in \{0,1,\cdots,N\}$. At each time, if the source node decides not to probe the channel state then it will not sample any process, thus $b(t)=0, a(t)=0$, and the expected single-stage AoI cost is $c(s(t),b(t), a(t))=\sum_{i=1}^N T_i$. However, if the source node decides to probe the channel state, the expected single-stage AoI cost is $c(s(t),b(t)=1, a(t)=0)=\sum_{i=1}^N T_i$ and $c(s(t), b(t)=1, a(t)=k)=\sum_{i \neq k} T_i+T_k(1-p(C))$ where the expectation is taken over packet success probability $p(C)$. 
\subsubsection{Optimality equation}
In this case, the Bellman equations are given by:

\scriptsize
\begin{eqnarray}\label{eqn:Bellman-eqn-single-sensor-multi-process-with-fading-general}
&&J^{*}(E \geq E_{p}+E_{s},T_1, T_2,\cdots, T_N)\nonumber\\
&=&min \bigg\{\sum_{i=1}^N T_i+\alpha \mathbb{E}_{A}J^{*}(min\{E+A,B\},T_1+1, T_2+1,\cdots, \nonumber\\
&&T_N+1), V^{*}(E,T_1, T_2,\cdots, T_N)\bigg\}\nonumber\\
&&V^{*}(E,T_1, T_2,\cdots, T_N)\nonumber\\
&=& \sum_{j=1}^m q_{j} W^{*}(E,T_1, T_2,\cdots, T_N,C_{j})\nonumber\\ 
&&W^{*}(E,T_1, T_2,\cdots, T_N,C)\nonumber\\ 
&=& min\{\sum_{i=1}^N T_i+\alpha \mathbb{E}_{A}J^{*}(min\{E-E_{p}+A,B\},T_1+1, T_2+1, \cdots, \nonumber\\ 
&& T_N+1),\min_{1 \leq k \leq N} \bigg( \sum_{i \neq k } T_i+T_k(1-p(C))+\alpha p(C)\mathbb{E}_{A}J^{*}(min\{E\nonumber\\
&&-E_{p}-E_{s}+A,B\},T_1+1,T_2+1, \cdots, T_k'=1,T_{k+1}+1,\cdots,\nonumber\\
&&T_N+1)+\alpha (1-p(C))\mathbb{E}_{A}J^{*}(min\{E-E_{p}-E_{s}+A,B\},T_1+1,\nonumber\\ 
&&T_2+1,\cdots, T_N+1)\bigg) \}\nonumber\\
&&J^{*}(E < E_{p}+E_{s},T_1, T_2,\cdots, T_N)\nonumber\\
&=&\sum_{i=1}^N T_i+\alpha \mathbb{E}_{A}J^{*}(min\{E+A,B\},T_1+1, T_2+1, \cdots,T_N+1) 
\end{eqnarray}
\normalsize

The first expression in the minimization in the R.H.S. of the first equation in \eqref{eqn:Bellman-eqn-single-sensor-multi-process-with-fading-general} is the cost of not probing channel state ($b(t)=0$), which includes single-stage AoI cost $\sum_{i=1}^N T_i$ and an $\alpha$ discounted future cost with a random next state $( \min\{E+A,B\},T_1+1, T_2+1, \cdots,T_N+1 )$, averaged over the distribution of the number of energy packet generation $A$. The quantity $V^{*}(E,T_1, T_2,\cdots, T_N)$ is the optimal expected cost of probing the channel state, which explains the second equation in
\eqref{eqn:Bellman-eqn-single-sensor-multi-process-with-fading-general}. At an intermediate state $(E,T_1, T_2,\cdots, T_N,C)$, if $a(t)=0$,
 a single stage AoI cost $\sum_{i=1}^N T_i$ is incurred and the next state becomes $(\min\{E-E_{p}+A,B\},T_1+1, T_2+1, \cdots,T_N+1)$; if $a(t)=k$, the expected AoI cost is $\sum_{i \neq k } T_i+T_k(1-p(C))$ (expectation taken over the packet success probability $p(C)$), and the next random state becomes $(\min\{E-E_{p}-E_{s}+A,B\},T_1+1,T_2+1, \cdots, T_k'=1,  T_{k+1}+1,\cdots,T_N+1)$ and $(\min\{E-E_{p}-E_{s}+A,B\},T_1+1,T_2+1, \cdots, T_N+1)$ if $r(t)=1$ and $r(t)=0$, respectively. The last equation in \eqref{eqn:Bellman-eqn-single-sensor-multi-process-with-fading-general} follows similarly since $b(t)=0, a(t)=0$ is the only possible action when $E<E_p+E_s$.

Substituting the value of $V^{*}(E,T_1, T_2,\cdots, T_N)$ in the first equation of \eqref{eqn:Bellman-eqn-single-sensor-multi-process-with-fading-general}, we obtain the Bellman equations
\eqref{eqn:Bellman-eqn-single-sensor-multiple-process-with-fading}.

\begin{figure*}
\scriptsize
\begin{eqnarray}\label{eqn:Bellman-eqn-single-sensor-multiple-process-with-fading}
J^{*}(E \geq E_{p}+E_{s},T_1, T_2,\cdots, T_N)&=&min \bigg\{\sum_{i=1}^N T_i+\alpha \mathbb{E}_{A}J^{*}(min\{E+A,B\},T_1+1, T_2+1, \cdots,T_N+1), \nonumber\\
&&\mathbb{E}_{C} \bigg( min\{\sum_{i=1}^N T_i+\alpha \mathbb{E}_{A}J^{*}(min\{E-E_{p}+A,B\},T_1+1, T_2+1, \cdots, T_N+1),\nonumber\\ 
&&\min_{1 \leq k \leq N} \bigg( \sum_{i \neq k } T_i+T_k(1-p(C))+\alpha p(C)\mathbb{E}_{A}J^{*}(min\{E-E_{p}-E_{s}+A,B\},T_1+1,\nonumber\\
&&T_2+1, \cdots, T_k'=1,T_{k+1}+1,\cdots,T_N+1)+\alpha (1-p(C))\mathbb{E}_{A}J^{*}(min\{E-E_{p}-E_{s}+A,B\},\nonumber\\
&&T_1+1,T_2+1, \cdots, T_N+1) \bigg)\} \bigg) \bigg\} \nonumber\\
J^{*}(E < E_{p}+E_{s},T_1, T_2,\cdots, T_N)&=&\sum_{i=1}^N T_i+\alpha \mathbb{E}_{A}J^{*}(min\{E+A,B\},T_1+1, T_2+1, \cdots,T_N+1) 
\end{eqnarray}
\hrule
\end{figure*}
\normalsize

\subsubsection{Policy structure}\label{subsubsection:Policy structure-IID-System-N>1}

\begin{lemma}\label{lemma:single-sensor-multiple-process-J-increasing-in-T-decreasing-in-p(C^')}
 For $N>1$, the value function $J^*(E, T_1,T_2, \cdots, T_N)$ is increasing in each of $T_1, T_2, \cdots, T_N $ and $W^{*}(E,T_1, T_2,\cdots, T_N,C)$ is decreasing in $p(C)$. 
 \end{lemma}
 
\begin{proof} See Appendix ~\ref{appendix:proof-of-lemma-single-sensor-multiple-process-fading-cost-increasing-in-T-decrease-in-p}.
\end{proof}
Let us define $\bm{T}=[T_1, T_2,\cdots, T_N]$ and $\bm{T}_{-k}=[T_1, T_2,\cdots, T_{k-1}, T_{k+1},\cdots, T_N]$. 

 \vspace{-6pt}
\begin{conjecture}\label{theorem:single-sensor-multiple-process-with-fading-policy-structure}
 For $N>1$, the optimal probing policy for the $\alpha$-discounted AoI cost minimization problem is a threshold policy on $ \max_{1 \leq k \leq N}T_{k} \doteq T_{k^*}$. For any $E \geq E_{p}+E_{s} $, the optimal action is to probe the channel state if and only if $ \max_{1 \leq k \leq N}T_{k}\geq T_{th}(E,\bm{T}_{-k^*})$ for a threshold function $T_{th}(E,\bm{T}_{-k^*})$ of $E$ and $\bm{T}_{-k^*}$.

 \end{conjecture}
 We next show that sampling the process with the largest AoI is optimal.
 
\begin{theorem}\label{theorem:single-sensor-multiple-process-with-fading-policy-structure}
 For $N>1$, after probing the channel state, the optimal source activation policy for the $\alpha$-discounted cost problem is a threshold policy on p(C). For any $E \geq E_{p}+E_{s} $ and probed channel state, the optimal action is to sample the process $\arg \max_{1 \leq k \leq N} T_k$ if and only if $p(C) \geq p_{th}(E, \bm{T})$ for a threshold function $p_{th}(E, \bm{T})$ of $(E, \bm{T})$.
 
\end{theorem}


\begin{proof} See Appendix ~\ref{appendix:proof-of-theorem-single-sensor-multiple-process-fading-threshold-policy}.
\end{proof}

 We will later numerically demonstrate some intuitive properties of $T_{th}(E,\bm{T}_{-k})$ as a function of $E$, $\bm{T}_{-k}$ and $\lambda$ and $p_{th}(E,T_1,T_2,\cdots,T_N)$ as a function of $E$, $(T_1,T_2,\cdots,T_N)$ and $\lambda$ in Section~\ref{section:numerical-work}.

  \vspace{-8pt}
 
 \section{Markovian System Model: single process }\label{section:Markovian channel and Markovian energy}
 
Here, we derive the optimal channel probing, source activation, and data transmission policy for an EH source sampling a single process with a Markovian channel and Markovian energy arrival process. 
We formulate the $\alpha$-discounted cost version of \eqref{eqn:main-problem} as an MDP with
state space $\mathcal{S} \doteq \{0,1,\cdots,B\} \times \mathbb{Z}_+ \times \mathbb{Z}_+ \times \{C_{1}, C_{2},\cdots, C_{m}\} \times \{H_{1}, H_{2}\}$ where a generic primary  state  
 $s = (E,T,\tau, C_{prev}, H_{1} )$ means that the energy buffer has $E$ energy packets, the last successfully received packet was generated at the source node $T$ slots ago, the channel state was last probed $\tau$ slots ago where $\tau \leq T$, the previous channel state obtained via probing is $C_{prev} \in \{C_1,C_2,\cdots,C_m\}$, and the EH source is in harvesting state $H_{1}$. A generic primary state $s = (E,T,\tau, C_{prev}, H_{2})$ similarly means that the EH source is in non-harvesting state $H_{2}$. Also, $\mathcal{V}$ denotes intermediate state space where a generic intermediate  state $v = (E,T, C, H_{1} )$ or $v = (E,T, C, H_{2} )$ is used for current channel state $C$ (obtained after probing). 
 The action space $\mathcal{A}=\{\{0,0\}\cup \{1\times\{0,1 \}\}\}$ with $a(t), b(t) \in \{0,1\}$. At each time, if $(b(t) = 0, a(t) = 0)$, then the expected single-stage AoI cost is $c(s(t), b(t), a(t)) = T$. However, if $(b(t) = 1)$, the expected single-stage AoI cost is $c(s(t), b(t)=1, a(t)=0) = T$, and $c(s(t), b(t)=1, a(t)=1) = T(1-p(C))$. 

 \vspace{-6pt}
\subsection{Optimality equation}
In this case, for  primary  state $(E,T,\tau, C_{prev}, H_{1})$, the Bellman equations are given by:

\scriptsize
\begin{eqnarray} \label{eqn:Bellman-eqn--single-process-with-markov}
&&J^{*}(E \geq E_{p}+E_{s},T,\tau, C_{prev}, H_{1})\nonumber\\
&=&min \bigg\{T+\alpha \mathbb{E}_{A|H_{1}, H_{next}|H_{1}}J^{*}(min\{E+A,B\},T+1, \tau +1, \nonumber\\
&& C_{prev},H_{next}), V^{*}(E,T,\tau,C_{prev}, H_{1}) \bigg\}\nonumber\\
&&V^{*}(E,T,\tau, C_{prev}, H_{1})\nonumber\\
&=& \mathbb{E}_{C\sim q^{(\tau)} (\cdot|C_{prev})} W^{*}(E,T, C, H_{1})\nonumber\\
&& W^{*}(E,T, C, H_{1})\nonumber\\
&=&min\{T+\alpha \mathbb{E}_{A|H_{1},H_{next}|H_{1}}J^{*}(min\{E-E_{p}+A,B\},T+1,1, C,\nonumber\\ 
&&H_{next}),T(1-p(C))+\alpha p(C)\mathbb{E}_{A|H_{1},H_{next}|H_{1}}J^{*}(min\{E-E_{p}-\nonumber\\
&&E_{s}+A,B\},1,1, C,H_{next})+\alpha (1-p(C))\mathbb{E}_{A|H_{1},H_{next}|H_{1}}J^{*}(\nonumber\\
&&min\{E- E_{p}-E_{s}+A,B\},T+1,1, C,H_{next})\}\nonumber\\
&&J^{*}(E< E_{p}+E_{s},T,\tau, C_{prev}, H_{1})\nonumber\\
&=&T+\alpha \mathbb{E}_{A|H_{1}, H_{next}|H_{1}}J^{*}(min\{E+A,B\},T+1,\tau +1, C_{prev},\nonumber\\
&&H_{next} )
\end{eqnarray}
\normalsize

The first expression in the minimization in the R.H.S. of the first equation in \eqref{eqn:Bellman-eqn--single-process-with-markov} is the cost when $b(t)=0$. The quantity $V^{*}(E,T,\tau,C_{prev}, H_{1})$ is the expected cost when $b(t)=1$, which explains the second equation in \eqref{eqn:Bellman-eqn--single-process-with-markov}. Here $q^{(\tau)}$ denote the $\tau$-step transition probability of the channel dynamics, starting from $C_{prev}$.

Similarly, in \eqref{eqn:Bellman-eqn--single-process-with-markov}, if we consider $H_{2}$ for primary  state $(E,T,\tau, C_{prev}, H_{2})$ then the same Bellman equations hold with $A=0$. 


\subsection{Policy structure}\label{subsection:policy structure_markov}
We conjecture the following properties regarding the value function and the optimal policy. These conjectures are supported by numerical results which are given in \ref{subsection:markovian-sensor-single-process}.
\begin{conjecture} \label{conjecture:single-sensor-single-process-markov-value-function}
For the Markovian system model and $N=1$, the value function 
$J^{*}(E, T,\tau, C_{prev},H )$ is increasing in $T$ and $W^{*}(E,T, C,H)$ is decreasing in $p(C)$. 
\end{conjecture}

\begin{conjecture}\label{conjecture:single-sensor-single-process-markov}
 For the Markovian system model and $N=1$, the optimal probing policy for the $\alpha$-discounted AoI cost minimization problem is a threshold policy on $T$. For any $E \geq E_{p}+E_{s}$, the optimal action is to probe the channel state if and only if $T\geq T_{th}(E,\tau, C_{prev},H)$ for a threshold function $T_{th}(E,\tau,C_{prev},H)$ of $E,\tau,C_{prev},H$. 
\end{conjecture}

\begin{conjecture}\label{theorem:single-sensor-single-process-with-fading-policy-structure_markov}
 For the Markovian system model and $N=1$, at any time, if the source decides to probe the channel, then the optimal sampling policy is a threshold policy on $p(C)$. For any $E \geq E_{p}+E_{s} $ and probed channel state, the optimal action is to sample the source node if and only if $p(C)\geq p_{th}(E,T,H)$ for a threshold function $p_{th}(E,T,H)$ of $E$, $T$ and $H$. 
\end{conjecture}


 \vspace{-6pt}
\section{Reinforcement Learning for AoI Optimization:   Single Process}\label{section:RL}
 If the channel statistics and the energy harvesting characteristics are not known at the source, then the source has to learn them with time. We adapt the Q-learning technique \cite[Section 11.4.2]{s2013stochastic} to find the optimal policy for our two-stage action model.

 \vspace{-6pt} 
 \subsection{I.I.D. system}\label{IID system-RL}
 Here, we assume that the energy generation rate $\lambda$ and the channel state probabilities $\{q_j\}_{1 \leq j \leq m}$ are not known to the source. The packet success probabilities $\{p_j\}_{1 \leq j \leq m}$ are either known or unknown; the packet success indicator $r(t)$ is used to estimate $p_j$ if it is unknown.\\
 
  \vspace{-10pt}
 \subsubsection{Optimality equation in terms of Q-function}
 \label{section:Optimality equation for Q-fuction SSSP}
 
The optimal Q-function is denoted by $Q^{*}(E,T,b)$ for a generic primary state-action pair $(E,T,b)$, where the primary state is $(E,T) $ and the action is  $b \in \{0,1\}$. The quantity $Q^{*}(E,T,b)$ denotes the expected cost obtained if the current primary state is $(E,T)$, the current action chosen is $b$, and an optimal policy is employed from the next time instant. Also,   $Q^{*}(E,T,C, a)$ corresponds to the   intermediate state $(E,T,C)$ and a corresponding  action   $a \in \{0,1\}$. It is important to note that, we have to maintain two different types of Q-functions to handle primary states and intermediate states; this is not done in standard Q-learning. 
 The optimal Q-value functions are given by:
 
 \scriptsize
\begin{eqnarray}\label{eqn:Q-eqn-single-process-1} 
Q^{*}(E \geq E_{p}+E_{s},T,b=0)&=&T+\alpha \mathbb{E}_{A}J^{*}(min\{E+A,B\},T+1) \nonumber \\
&=& T+\alpha \mathbb{E}_{A}\min_{b \in \{ 0,1 \}}Q^{*}(min\{E+A,B\},\nonumber\\
&&T+1,b)
\end{eqnarray}
\normalsize

Here, the terms in the R.H.S of the first equality operator in   \eqref{eqn:Q-eqn-single-process-1} is the immediate cost of not probing channel state action $(b=0)$ when primary  state is $(E \geq E_{p}+E_{s},T)$ and an $\alpha$ discounted future cost $\mathbb{E}_{A}J^{*}(min\{E+A,B\},T+1)$  which is later substituted by 
$\mathbb{E}_{A}\min_{b \in \{ 0,1 \}}Q^{*}(min\{E+A,B\},T+1,b)$ (averaged over the distribution of A). On the other hand, 

\scriptsize
\begin{eqnarray} \label{eqn:Q-eqn-single-process-2}
Q^{*}(E \geq E_{p}+E_{s},T,b=1)  
&=& \sum_{j=1}^m q_{j} W^{*}(E,T,C_{j})\nonumber\\ 
&=&\sum_{j=1}^m q_{j}\min_{a \in \{ 0,1 \}}Q^{*}(E,T,C_{j}, a)
\end{eqnarray}
\normalsize

Now, 

\scriptsize
\begin{eqnarray} \label{eqn:Q-eqn-single-process-4}
&&Q^{*}(E,T,C,a=0)\nonumber\\
&=&T+\alpha \mathbb{E}_{A}\underbrace{\min_{b \in \{ 0,1 \}}Q^{*}(min\{E-E_{p}+A,B\}, T+1,b)}_{=J^*(min\{E-E_{p}+A,B\}, T+1)}
\end{eqnarray}
\normalsize

and 

\scriptsize
\begin{eqnarray} \label{eqn:Q-eqn-single-process-5}
&&Q^{*}(E,T,C,a=1)\nonumber\\
&=&T(1-p(C))+\alpha p(C)\mathbb{E}_{A}\min_{b \in \{ 0,1 \}}Q^{*}(min\{E-E_{p}-E_{s}+A,B\},\nonumber\\
&&1,b)+\alpha (1-p(C))\mathbb{E}_{A}\min_{b \in \{ 0,1 \}}Q^{*}(min\{E-E_{p}-E_{s}+A,B\},\nonumber\\
&&T+1,b)
\end{eqnarray}
\normalsize

Similarly, the optimal Q-function for primary state $(E< E_{p}+E_{s},T)$ is given by \eqref{eqn:Q-eqn-single-sensor-single-process-with-fading-general} in which only not probing  $(b=0)$ is a feasible action:

\scriptsize
\begin{eqnarray} \label{eqn:Q-eqn-single-sensor-single-process-with-fading-general}
Q^{*}(E< E_{p}+E_{s},T,b=0)&=&T+\alpha \mathbb{E}_{A}J^{*}(min\{E+A,B\},T+1)\nonumber\\
&=&T+\alpha \mathbb{E}_{A}\min_{b \in \{ 0,1 \}}Q^{*}(min\{E+A,B\},\nonumber\\
&&T+1,b)
\end{eqnarray}

\normalsize

\subsubsection{Q-Learning algorithm}\label{section:Optimality equation for Q-fuction SSSP}
The optimal Q-function is given by the zeros of \eqref{eqn:Q-eqn-single-process-1}-\eqref{eqn:Q-eqn-single-sensor-single-process-with-fading-general}. Here, we employ asynchronous stochastic approximation \cite{borkar1998asynchronous} to iteratively converge to the optimal Q-function, starting from any arbitrary Q-function. Online learning is facilitated by sequentially observing $C(t)$ (via channel probing), $r(t)$ (via transmission ACK/NACK), and the number $A(t)$ of energy packets harvested at time~$t$.

In asynchronous stochastic approximation, we need step size sequence $d(t)$ ($t \geq 0$) which  satisfies the following assumptions \cite{bhatnagar2011borkar}:  
\begin{assumption}\label{assumption2}
(i) $0< d(t) \leq \bar{d}$ where $\bar{d} > 0$, \\
(ii) $d(t+1)\leq d(t)$ for all $t \geq 0$, \\
(iii) $\sum_{t=0}^\infty d(t) = \infty$ and $\sum_{t=0}^\infty (d(t))^{2} < \infty$.\\
\end{assumption}
Let us denote by  $\nu_{t}(s,b)$   the number of occurrences of the primary state-action pair $(s,b)$ up to iteration $t$ where $s \in \mathcal{S} $ and $b \in \{0,1\}$. A similar notation is assumed for any generic intermediate state $v$ and the corresponding action $a$.   The following assumption is required for the convergence of asynchronous stochastic approximation \cite{bhatnagar2011borkar}:
\begin{assumption}\label{assumption1}
   $\lim \inf_{t \rightarrow \infty} \frac{\nu_{t}(s,b)}{t} > 0$ almost surely $\forall (s,b)$, and $\lim \inf_{t \rightarrow \infty} \frac{\nu_{t}(v,a)}{t} > 0$ almost surely $\forall (v,a)$.
 \end{assumption}

The proposed algorithm maintains a look-up table $Q_t(\cdot,\cdot)$ for various state-action pairs, and iteratively updates its entries depending on the current state, the current action taken and the observed next state.
 However, Assumption~\ref{assumption1} requires that all state-action pairs should be visited infinitely and comparatively often. This is ensured by taking a random action (uniformly chosen) with probability $\epsilon$ at each decision instant, and taking the action $\arg min_b Q(s,b)$ or $\arg min_a Q(v,a)$ with probability $(1-\epsilon)$.

One example of   a Q-value update  is the following, aimed at convergence to the solution of \eqref{eqn:Q-eqn-single-process-1}:

\scriptsize
\begin{eqnarray} \label{eqn:Q-iteration-single-process-iterate1}
&&Q_{t+1}(E \geq E_{p}+E_{s},T,b=0)\nonumber\\
&=&Q_{t}(E,T,b=0)+d(\nu_{t}(E,T,b=0))\mathrm{\mathbf{1}}\{s(t)=(E,T),b(t)=0\}[ \nonumber\\
&&T+\alpha \min_{b \in \{ 0,1 \}}Q_{t}(min\{E+A(t),B\},T+1,b)-Q_{t}(E,T,b=0)]\nonumber\\
\end{eqnarray}
\normalsize

Here,  $\mathrm{\mathbf{1}}\{\cdot\}$ is the indicator function. It is important to note that,  this $Q$-update can be performed only when the random next state $(min\{E+A(t),B\},T+1)$ is observed. 

Similarly,  \eqref{eqn:Q-eqn-single-process-2}-\eqref{eqn:Q-eqn-single-sensor-single-process-with-fading-general} lead to the following Q-updates for various primary states and intermediate states:

\scriptsize
\begin{eqnarray} \label{eqn:Q-iteration-single-process-iterate2}
&&Q_{t+1}(E \geq E_{p}+E_{s},T,b=1)\nonumber\\
&=& Q_{t}(E,T,b=1)+d(\nu_{t}(E,T,b=1)\mathrm{\mathbf{1}}\{s(t)=(E,T),b(t)=1\}\nonumber\\
&&[\min_{a\in \{ 0,1 \}}Q_{t}(E,T,C, a)-Q_{t}(E,T,b=1)]
\end{eqnarray}
\normalsize

\scriptsize
\begin{eqnarray} \label{eqn:Q-iteration-single-process-iterate4}
&&Q_{t+1}(E,T,C,a=0)\nonumber\\
&=&Q_{t}(E,T,C,a=0)+d(\nu_{t}(E,T,C,a=0))\mathrm{\mathbf{1}}\{v(t)=(E,T,C),\nonumber\\
&&a(t)=0\}[T+\alpha \min_{b\in \{ 0,1 \}}Q_{t}(min\{E-E_{p}+A(t),B\},T+1,b)-\nonumber\\
&&Q_{t}(E,T,C,a=0)]
\end{eqnarray}
\normalsize

\scriptsize
\begin{eqnarray} \label{eqn:Q-iteration-single-process-iterate4.1}
&&Q_{t+1}(E,T,C,a=1)\nonumber\\
&=&Q_{t}(E,T,C,a=1)+d(\nu_{t}(E,T,C,a=1))\mathrm{\mathbf{1}}\{v(t)=(E,T,C),\nonumber\\
&&a(t)=1\}[T(1-p(C))+\alpha p(C)\min_{b\in \{ 0,1 \}}Q_{t}(min\{E-E_{p}-E_{s}+\nonumber\\
&&A(t),B\},1,b)+\alpha (1-p(C))\min_{b\in \{ 0,1 \}}Q_{t}(min\{E-E_{p}-E_{s}+A(t),\nonumber\\
&&B\},T+1,b)-Q_{t}(E,T,C,a=1)]
\end{eqnarray}
\normalsize

\scriptsize
\begin{eqnarray} \label{eqn:Q-iteration-single-process-iterate5}
&&Q_{t+1}(E< E_{p}+E_{s},T,b=0)\nonumber\\
&=&Q_{t}(E,T,b=0)+d(\nu_{t}(E,T,b=0))\mathrm{\mathbf{1}}\{s(t)=(E,T),b(t)=0\}[T+ \nonumber\\
&&\alpha \min_{b\in \{ 0,1 \}}Q_{t}(min\{E+A(t),B\},T+1,b)-Q_{t}(E,T,b=0)]
\end{eqnarray}
\normalsize

In  \eqref{eqn:Q-iteration-single-process-iterate4.1}, we consider the case where packet success probability $p(C)$ for channel state $C$ is known to the source. However, if these probabilities are unknown, then one can replace $p(C)$ in \eqref{eqn:Q-iteration-single-process-iterate4.1} by $r(t)$ to obtain a valid Q-update.

Q-learning type algorithm for $N>1$ can also be formulated similarly.

 \vspace{-10pt} 
 
 \subsection{Markovian system}
\label{section:RL_Markovian}
 Here, we adapt Q-learning for the Markovian model with unknown channel statistics and unknown energy harvesting characteristics to obtain the optimal solutions for \eqref{eqn:Bellman-eqn--single-process-with-markov}. 
 \subsubsection{Optimality equation in terms of Q-function}\label{section:Optimality equation for Q-fuction SSSP-markov}
 The optimal Q-value functions are given by:
 
 \scriptsize
\begin{eqnarray}\label{eqn:Q-eqn-single-process-markov-1} 
&&Q^{*}(E \geq E_{p}+E_{s},T,\tau,C_{prev}, H_{1},b=0)\nonumber\\
&=&T+\alpha \mathbb{E}_{A|H_{1}, H_{next}|H_{1}}J^{*}(min\{E+A,B\},T+1,\tau+1, C_{prev},\nonumber\\
&&H_{next})\nonumber\\
&=&T+\alpha \mathbb{E}_{A|H_{1}, H_{next}|H_{1}}\min_{b\in \{ 0,1 \}}Q^{*}(min\{E+A,B\},T+1,\tau+1,\nonumber\\
&&C_{prev},H_{next}, b)
\end{eqnarray}
\normalsize

On the other hand, 

\scriptsize
\begin{eqnarray} \label{eqn:Q-eqn-single-process-markov-2}
&&Q^{*}(E \geq E_{p}+E_{s},T,\tau,C_{prev}, H_{1},b=1)\nonumber\\
&=& \mathbb{E}_{C\sim q^{(\tau)} (\cdot|C_{prev})}W^{*}(E,T,C, H_{1})\nonumber\\ 
&=& \mathbb{E}_{C\sim q^{(\tau)} (\cdot|C_{prev})}\min_{a \in \{ 0,1 \}}Q^{*}(E,T,C,H_{1}, a)
\end{eqnarray}
\normalsize

When $b(t)=1$, for a current measured channel state $C$, the optimal ${Q}^{*}$-functions for action $a=0$ and $a=1$ are given by:

\scriptsize
\begin{eqnarray} \label{eqn:Q-eqn-single-process-markov-4}
&&Q^{*}(E,T,C,H_{1}, a=0)\nonumber\\
&=&T+\alpha \mathbb{E}_{A|H_{1},H_{next}|H_{1}}\min_{b\in \{ 0,1 \}}Q^{*}(min\{E-E_{p}+A,B\},T+1,\nonumber\\
&& 1,C,H_{next},b)
\end{eqnarray}
\normalsize
\scriptsize
\begin{eqnarray} \label{eqn:Q-eqn-single-process-markov-5}
&&Q^{*}(E,T,C,H_{1}, a=1)\nonumber\\
&=&T(1-p(C))+\alpha p(C)\mathbb{E}_{A|H_{1},H_{next}|H_{1}}\min_{b\in \{ 0,1 \}}Q^{*}(min\{E-E_{p}-\nonumber\\
&&E_{s}+A,B\},1,1, C,H_{next},b)+\alpha (1-p(C))\mathbb{E}_{A|H_{1},H_{next}|H_{1}}\nonumber\\
&&\min_{b\in \{ 0,1 \}}Q^{*}(min\{E-E_{p}-E_{s}+ A,B\},T+1,1, C,H_{next},b)\}
\end{eqnarray}
\normalsize

Also, the optimal Q-function for primary state $(E< E_{p}+E_{s},T,\tau,C_{prev}, H_{1})$ is given by:

\scriptsize
\begin{eqnarray} \label{eqn:Q-eqn-single-sensor-markov-general}
&&Q^{*}(E< E_{p}+E_{s},T,\tau,C_{prev}, H_{1},b=0)\nonumber\\
&=&T+\alpha \mathbb{E}_{A|H_{1}, H_{next}| H_{1}}\min_{b\in \{ 0,1 \}}Q^{*}(min\{E+A,B\},T+1,\tau+1,\nonumber\\
&&C_{prev},H_{next}, b)
\end{eqnarray}
\normalsize

Similarly, the Q-functions are given for primary state $(E,T,\tau, C_{prev}, H_{2})$ by putting $A=0$. 

\subsubsection{Q-Learning algorithm}\label{section:iteration update for Q-fuction SSSP markov}
 Here also, we employ asynchronous stochastic approximation to find the solution of \eqref{eqn:Q-eqn-single-process-markov-1}-\eqref{eqn:Q-eqn-single-sensor-markov-general} as in  \ref{IID system-RL}. The Q-updates for various primary states and intermediate states are given by the following:

\scriptsize
\begin{eqnarray} \label{eqn:Q-iteration-single-process-markov-iterate1}
&&Q_{t+1}(E \geq E_{p}+E_{s},T,\tau, C_{prev}, H_{1},b=0)\nonumber\\
&=&Q_{t}(E,T,\tau,C_{prev}, H_{1},b=0)+d(\nu_{t}(E,T,\tau,C_{prev}, H_{1},b=0))\nonumber\\
&&\mathrm{\mathbf{1}}\{s(t)=(E,T,\tau,C_{prev}, H_{1}),b(t)=0\}[T+ \alpha \min_{b \in \{ 0,1 \}}Q_{t}(min\{E\nonumber\\
&&+A(t),B\},T+1,\tau+1,C_{prev}, H_{next}, b)-Q_{t}(E,T,\tau,C_{prev}, H_{1},\nonumber\\
&&b=0)]
\end{eqnarray}
\normalsize


\scriptsize
\begin{eqnarray} \label{eqn:Q-iteration-single-process-markov-iterate2}
&&Q_{t+1}(E \geq E_{p}+E_{s},T,\tau,C_{prev}, H_{1},b=1)\nonumber\\
&=&Q_{t}(E,T,\tau,C_{prev}, H_{1},b=1)+d(\nu_{t}(E,T,\tau,C_{prev}, H_{1},b=1)\nonumber\\
&&\mathrm{\mathbf{1}}\{s(t)=(E,T,\tau,C_{prev}, H_{1}),b(t)=1\}[\min_{a\in \{ 0,1 \}}Q_{t}(E,T,C,H_{1}, a)\nonumber\\
&&-Q_{t}(E,T,\tau,C_{prev}, H_{1},b=1)]
\end{eqnarray}
\normalsize

\scriptsize
\begin{eqnarray} \label{eqn:Q-iteration-single-process-markov-iterate4}
&&Q_{t+1}(E,T,C,H_{1},a=0)\nonumber\\
&=&Q_{t}(E,T,C,H_{1},a=0)+d(\nu_{t}(E,T,C,H_{1},a=0))\mathrm{\mathbf{1}}\{v(t)=\nonumber\\
&&(E,T,C,H_{1}),a(t)=0\}[T+\alpha \min_{b\in \{ 0,1 \}}Q_{t}(min\{E-E_{p}+A(t),B\},\nonumber\\
&&T+1,1,C,H_{next}, b)-Q_{t}(E,T,C,H_{1},a=0)]
\end{eqnarray}
\normalsize
\scriptsize
\begin{eqnarray} \label{eqn:Q-iteration-single-process-markov-iterate4-1}
&&Q_{t+1}(E,T,C,H_{1},a=1)\nonumber\\
&=&Q_{t}(E,T,C,H_{1},a=1)+d(\nu_{t}(E,T,C,H_{1},a=1))\mathrm{\mathbf{1}}\{v(t)=\nonumber\\
&&(E,T,C,H_{1}),a(t)=1\}[T(1-p(C))+\alpha p(C)\min_{b\in \{ 0,1 \}}Q_{t}(min\{E\nonumber\\
&&-E_{p}-E_{s}+A(t),B\},1,1,C,H_{next},b)+\alpha (1-p(C))\min_{b\in \{ 0,1 \}}Q_{t}(\nonumber\\
&&min\{E-E_{p}-E_{s}+A(t),B\},T+1,1,C,H_{next},b)-\nonumber\\
&&Q_{t}(E,T,C,H_{1},a=1)] 
\end{eqnarray}
\normalsize

\scriptsize
\begin{eqnarray} \label{eqn:Q-iteration-single-process-markov-iterate5}
&&Q_{t+1}(E< E_{p}+E_{s},T,\tau,C_{prev}, H_{1},b=0)\nonumber\\
&=&Q_{t}(E,T,\tau,C_{prev}, H_{1},b=0)+d(\nu_{t}(E,T,\tau,C_{prev}, H_{1},b=0))\nonumber\\
&&\mathrm{\mathbf{1}}\{s(t)=(E,T,\tau,C_{prev}, H_{1}),b(t)=0\}[T+\alpha \min_{b\in \{ 0,1 \}}Q_{t}(min\{E\nonumber\\
&&+A(t), B\},T+1,\tau+1,C_{prev},H_{next},b)-Q_{t}(E,T,\tau,C_{prev}, H_{1},\nonumber\\
&&b=0)]
\end{eqnarray}
\normalsize

Similarly, the Q-value function iteration updates are given for primary state $(E,T,\tau, C_{prev}, H_{2})$ by taking $A=0$. Also, when $p(C)$ is not known, it can be replaced by $r(t)$ in \eqref{eqn:Q-iteration-single-process-markov-iterate4-1}.

 \vspace{-8pt}

\section{Numerical results}\label{section:numerical-work}

{\bf Note:} For simulation purpose, we assume that the maximum age of a process(es) is upper bounded by $T_{max}=100$.

\subsection{ Single process, I.I.D. system, known dynamics }\label{subsection:single-sensor-single=process-channel-fading}
We consider five channel states ($m=5$) with channel state occurrence probabilities $\bm{q}=[0.2,0.2,0.2,0.2,0.2]$ and the corresponding packet success probabilities $\bm{p}=[0.9, 0.7, 0.5, 0.3, 0.1]$. Energy arrival process is i.i.d. $Bernoulli(\lambda)$ with energy buffer size $B=12$, $E_{p}=1$ unit, $E_{s}=1$ unit \cite{hatami2021aoi}.  Numerical exploration reveals that there exists a threshold policy on $T$ in decision-making for channel state probing; see Fig. ~\ref{fig2}(a) which substantiates our conjecture in \ref{subsubsection:Policy structure-IID-System-N=1}. It is observed that this $T_{th}(E)$ decreases with $E$ since higher available energy in the energy buffer allows the EH node to probe the channel state more aggressively. Similar reasoning explains the observation that $T_{th}(E)$ decreases with $\lambda$. For probed channel state, Fig.~\ref{fig2}(b) shows the variation of $p_{th}(E,T)$ with $E,T, \lambda$. It is observed that $p_{th}(E,T)$ decreases with $E$ since the EH node tries to sample the process more aggressively if more energy is available in the buffer. Similarly, a higher value of $T$ results in aggressive sampling, and hence $p_{th}(E,T)$ decreases with $T$. By similar arguments as before, we can explain the observation that this $p_{th}(E,T)$ decreases with $\lambda$. 

Further, rigorous numerical analysis also reveals that the optimal decision rules for probing and sampling are threshold rules on $E$, which we omit in this paper due to lack of space. 

\vspace{-8pt}
\begin{figure}[h]
  \begin{center}
 \subfloat[]{\includegraphics[height=4cm,width=8cm]{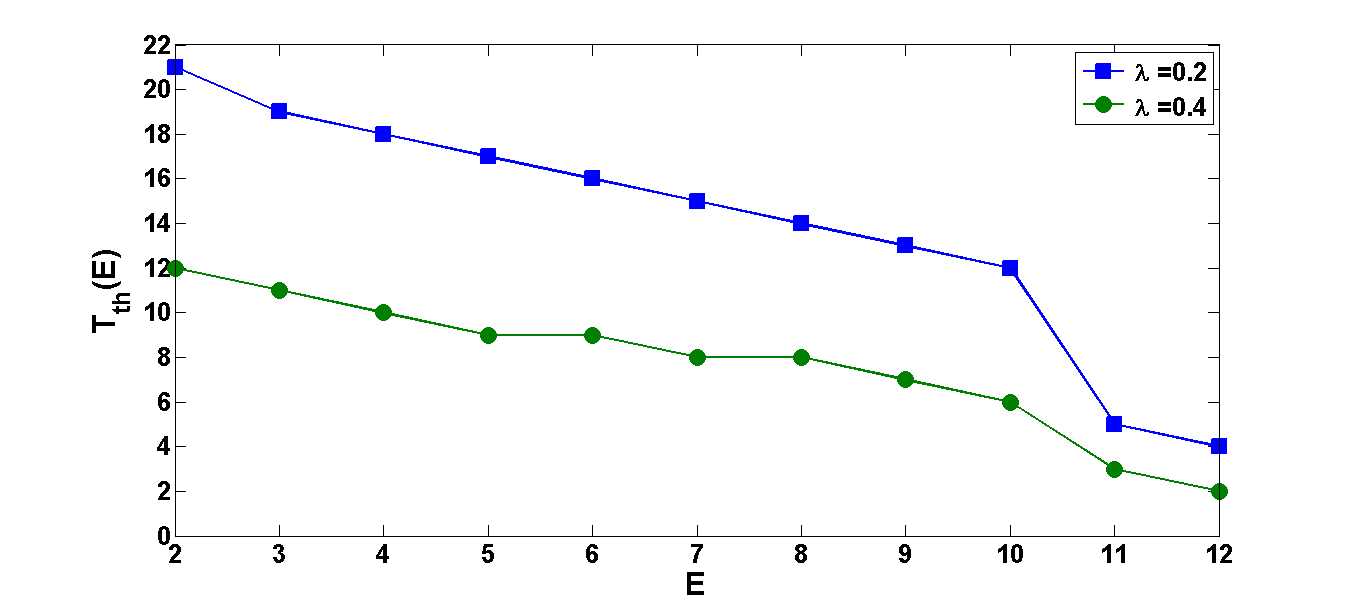}}
   \hfill
  \subfloat[]{\includegraphics[height=4cm,width=8cm]{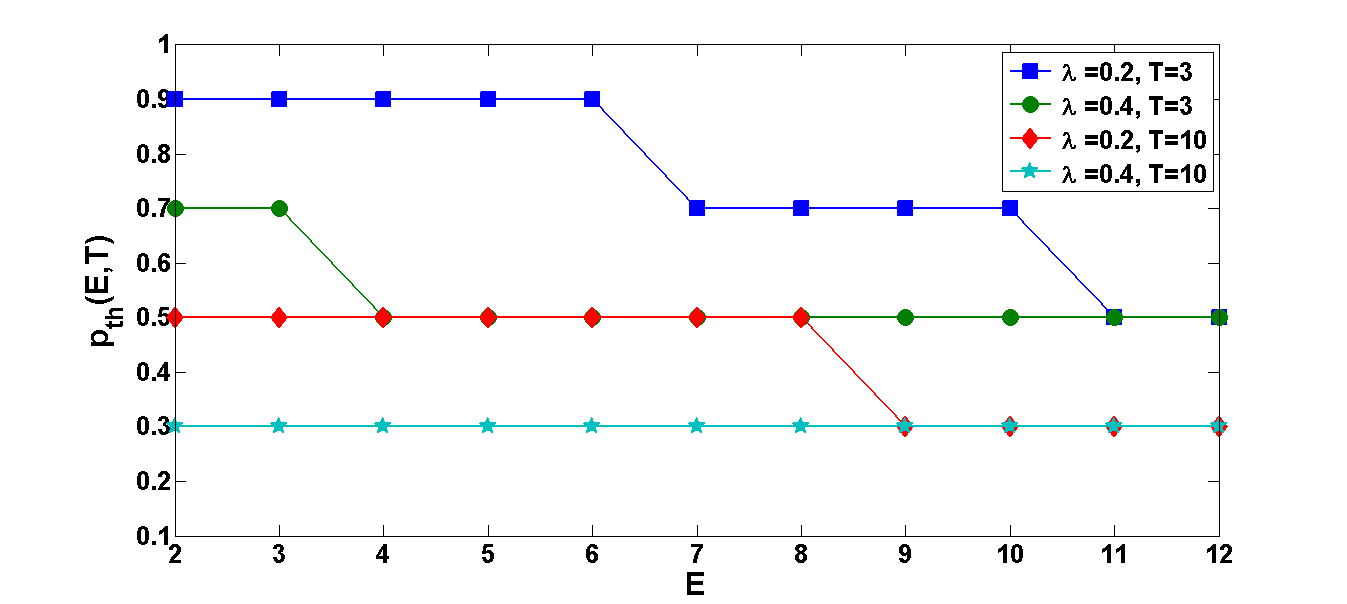}}
 \caption{Single process ($N=1$), I.I.D. model, known dynamics: (a) Variation of $T_{th}(E)$ with $E, \lambda$ and (b) Variation of $p_{th}(E,T)$ with $E, T, \lambda$.}
 \label{fig2}
 \end{center}
 \vspace{-18pt}
 \end{figure}

\subsection{Multiple processes ($N>1$), I.I.D. model, known dynamics}
We choose $N=3,   B=12, E_{p}=1, E_{s}=1$, and the same channel model and parameters as in Section~\ref{subsection:single-sensor-single=process-channel-fading}. Fig.~\ref{fig2.1}(a) provides support to our conjecture for optimal probing policy in \ref{subsubsection:Policy structure-IID-System-N>1} by demonstrating the variation on the threshold on $T_1$ for given $T_2, T_3$, for channel state probing. It is observed that $T_{th}(E, T_2, T_3 )$ decreases with $E$ and $\lambda$. Extensive numerical work also demonstrated that for the same   $\lambda$, this threshold decreases with each of $T_2$, $T_3$, since an increase in the age of processes results in more frequent channel probing. For probed channel state, Fig.~\ref{fig2.1}(b) shows that $p_{th}(E,T_1, T_2, T_3 )$ decreases with $E$ and $\lambda$. Further, the numerical analysis also demonstrated that for the same $\lambda$, this threshold decreases with each of $T_1$, $T_2$, $T_3$, since the source node becomes more aggressive in sampling. 

\begin{figure}[h]
  \begin{center}
 \subfloat[]{\includegraphics[height=4cm,width=7cm]{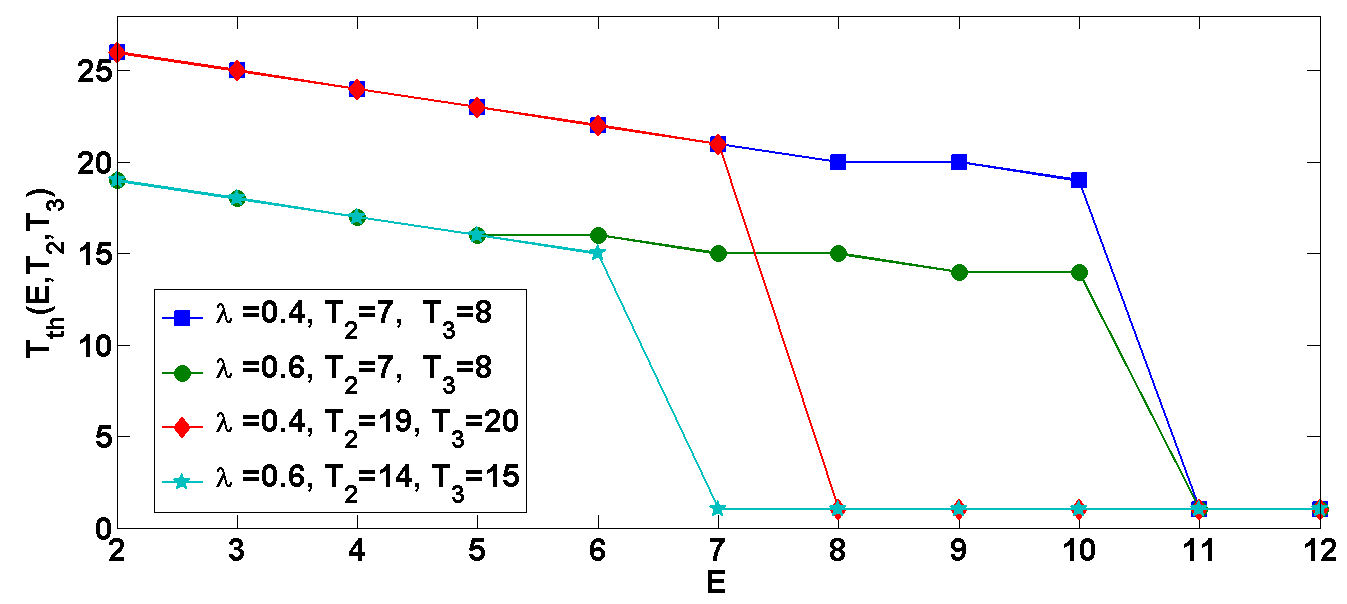}}
 \hfill
 \subfloat[]{ \includegraphics[height=4cm,width=7cm]{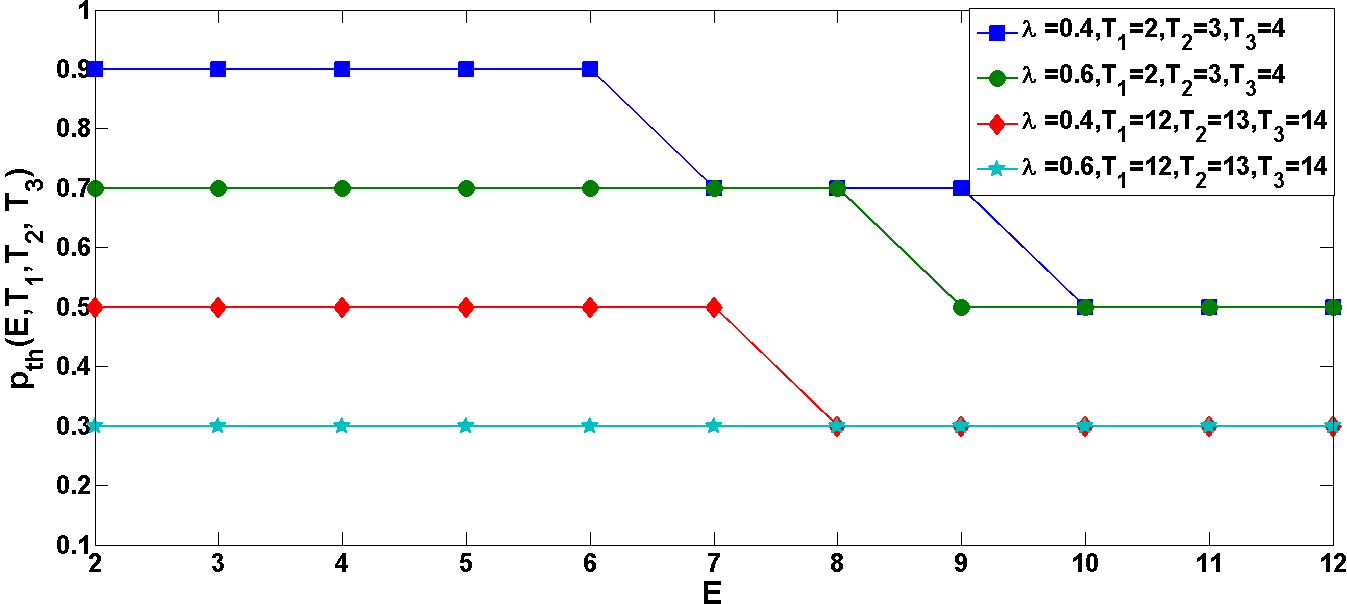}}
 \caption{$N=3$, I.I.D. model, known dynamics: (a) Variation of $T_{th}(E,T_2, T_3)$ with $E, T_2,T_3, \lambda$ and (b) Variation of $p_{th}(E,T_1, T_2, T_3)$ with $E,T_1,T_2, T_3, \lambda$.}
 \label{fig2.1}
 \end{center}
 \vspace{-15pt}
\end{figure}

\subsection{Markovian model, $N=1$, known dynamics }\label{subsection:markovian-sensor-single-process}

For numerical work, we have considered a two-state   ($m=2$) Gilbert-Eliot channel model  \cite{gilbert1960capacity, mushkin1989capacity} since it is practical as well as results in low computational complexity, and has been used in many recent papers involving AoI and EH source \cite{pan2022age, fountoulakis2023scheduling, yao2021age, abad2017channel}. While the Gilbert-Eliot channel model is a classic model, it is also highly relevant to millimeter-wave (mm-wave) communication which is a key driver for 5G and beyond wireless communications. Due to small wavelength, narrow beamwidth, and large propagation loss, mm-wave signals are highly susceptible to blocking and the channel can rapidly switch between on and off states due to the movement of obstacles and reflectors \cite{pan2022age, yao2019integrating, bildea2015link, he2014adaptive, dias2023millimeter}.  Thus, the Gilbert-Eliot channel model can be used for mm-wave communication channels with the ``On"  and ``Off" states representing Line of Sight (LOS) and Non-Line of Sight conditions, respectively. However, the procedure described in this paper can be used to numerically analyze any finite $m$-state channel model, and similar results can be obtained there too. \\
Here, we consider a two-state Markovian fading channel ($m=2$), where $C_{1}$ is the good channel state and $C_{2}$ is the bad channel state. The channel state transition probabilities are $q_{1,2}=0.1$ and $q_{2,1}=0.1$ \cite{fountoulakis2023scheduling, yao2021age}, and the corresponding packet transmission success probabilities are   $p_1=0.9$ and  $p_2=0.4$. Also, for the Markovian energy harvesting process, we consider the transition probabilities $h_{1,2}=0.3$ and $h_{2,1} =0.3$ \cite{sombabu2020age, ceran2021learning}. When the source node is in a harvesting state, the energy packet generation process is i.i.d. $Bernoulli(\lambda)$ across time, and otherwise zero. We consider   $B=9$, $E_{p}=1$ unit, $E_{s}=1$ unit, and $\lambda=0.4$.

Our numerical exploration demonstrated the threshold nature of the optimal policy which supports conjectures predicted in \ref{subsection:policy structure_markov}. Fig.~\ref{fig3} shows that, for a fixed $\tau$, $T_{th}(E,\tau,C_{prev},H)$ decreases with $E$, $T_{th}(E,\tau,C_{prev},H_1) \leq T_{th}(E,\tau,C_{prev},H_2)$   and $T_{th}(E,\tau,C_1,H) \leq T_{th}(E,\tau,C_2,H)$ (because probability of packet success in channel state $C_{1}$ is higher than that in $C_{2}$).   Also, from Fig.~\ref{fig3.1} it is observed that for probed channel state, $p_{th}(E,T,H)$ decreases with $E$ and $T$, and $p_{th}(E,T,H_1) \leq p_{th}(E,T,H_2)$.

\begin{figure}[h]
  \begin{center}
 \subfloat[]{\includegraphics[height=4cm,width=8cm]{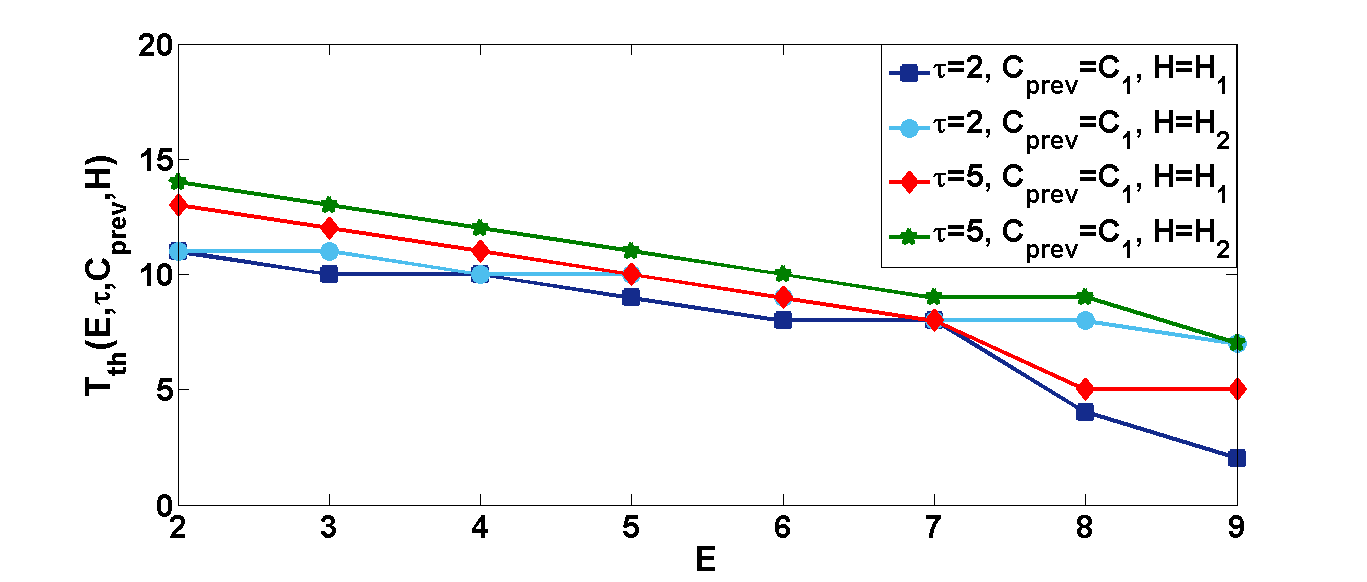}}
   \hfill
  \subfloat[]{\includegraphics[height=4cm,width=8cm]{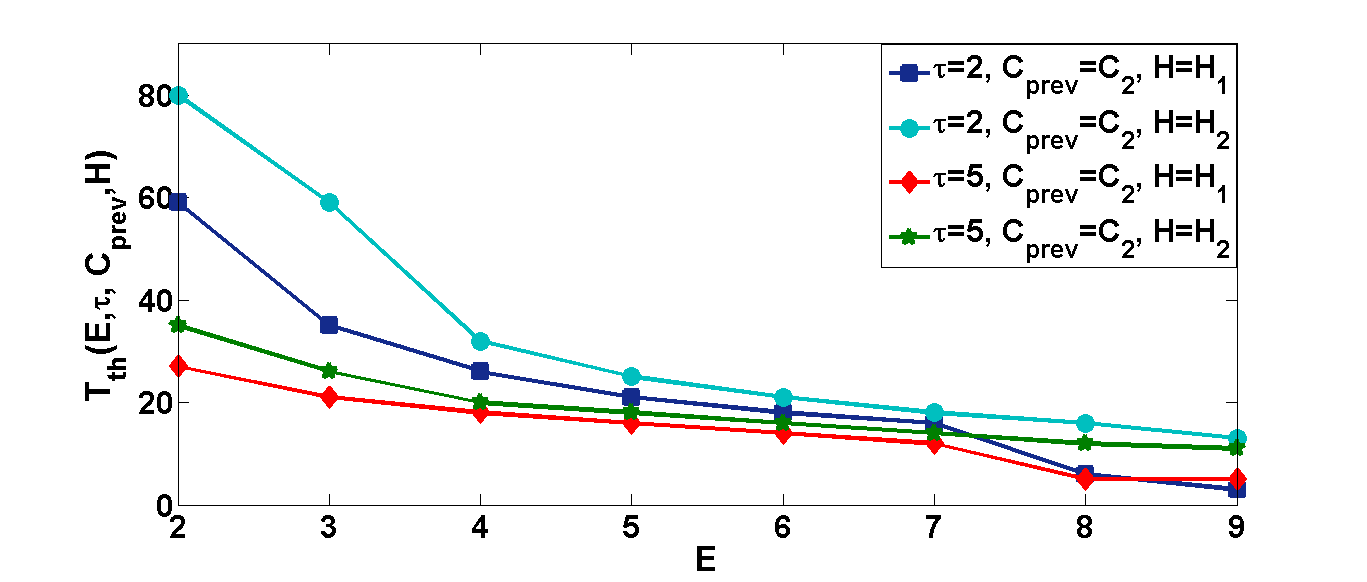}}
 \caption{ $N=1$, Markov model: (a) Variation of $T_{th}(E,\tau, C_{prev},H)$ with $E$, $\tau$, $C_{prev}=C_{1}$ and $H$  and 
 (b) Variation of $T_{th}(E,\tau, C_{prev},H)$ with $E$, $\tau$, $C_{prev}=C_{2}$ and $H$.}
 \label{fig3}
 \end{center}
\end{figure}

\begin{figure}[h]
  \begin{center}
 \includegraphics[height=4cm,width=8cm]{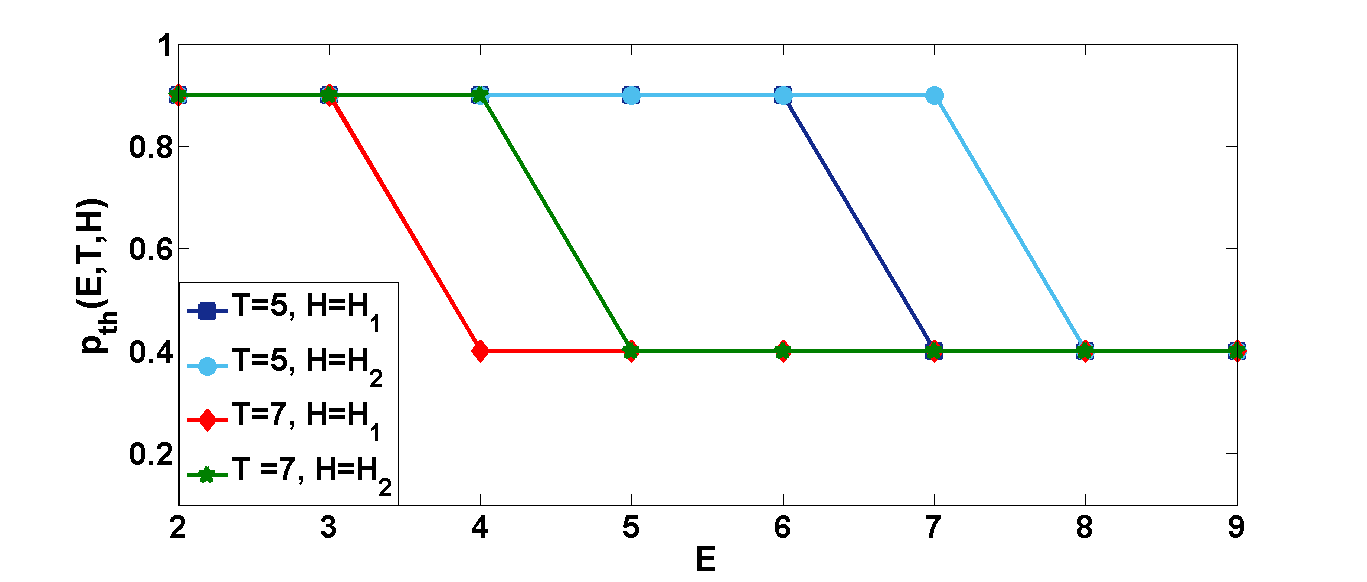}
 \caption{ $N=1$, Markov model: Variation of $p_{th}(E,T,H)$ with $E$, $T$, and $H$.}
 \label{fig3.1}
 \end{center}
 \vspace{-18pt}
 \end{figure}

Also, numerical exploration revealed that, after probing is done, the optimal sampling policy can also be represented as a threshold policy on $T$ instead of $p(C)$; in this case, sampling is done when $T \geq T_{th}(E,C,H)$. The variation of $  T_{th}(E,C,H)$ with $E,C,H$ is shown in Fig. ~\ref{fig4}, which supports many intuitive justifications provided earlier in this section. Similar result as in Fig. ~\ref{fig4} can also be obtained for the i.i.d. system model which we omit due to space constraint.

\begin{figure}[h]
  \begin{center}
 \includegraphics[height=4cm,width=8cm]{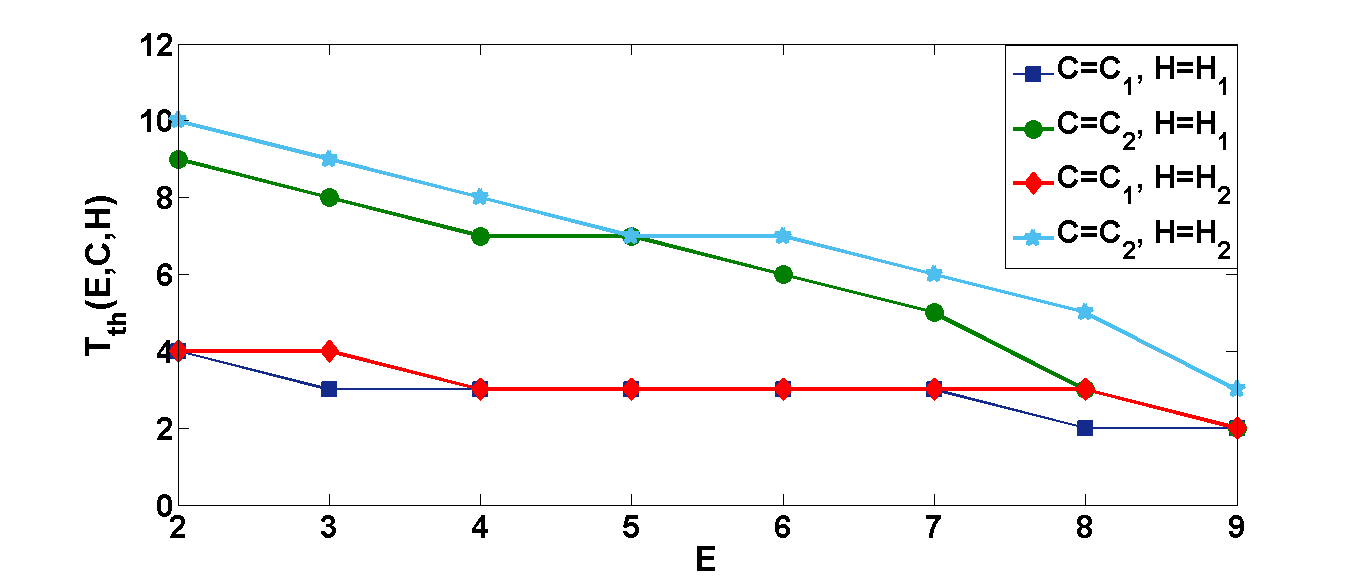}
 \caption{For $N=1$, Markov model: Variation of $T_{th}(E,C,H)$ with $E$, $C$, and $H$.}
 \label{fig4}
 \end{center}
 \end{figure}

\vspace{-8pt}
\subsection{ Comparison between probing and not probing }
\subsubsection{I.I.D. System}
We consider $N=1$ and the same model as in  Section~\ref{subsection:single-sensor-single=process-channel-fading}, and compare the performance of our probing-based threshold policy against the optimal policy that involves sampling mandatorily without probing (as in Section \ref{No probing case}) based on $T$ and $E$. The results are summarised in Fig. \ref{fig5} (a); we note that AoI increases with $E_p$ and $E_s$ since each action becomes more costly for the source. One interesting observation is that when the sensing energy is close to probing energy ($E_s=1,E_p=1$), we see slightly higher time-averaged  AoI values for the probing case compared to the case with $E_s=1$ and no probing. This happens because the advantage offered by probing is being overshadowed by the extra amount of energy expended for it. However,  when $E_s=5,E_p=1$, we notice that probing yields a lower time-averaged AoI. Obviously, the importance of probing is dependent on the trade-off between the advantage offered by probing and the extra energy required for it. Typically,  $E_p$ is much smaller than   $E_s$ since probing requires less energy as compared to sampling and transmitting a data packet. We have also observed numerically that, when there is sufficient variation in the channel quality, probing improves the AoI performance. On the other hand, when channel variation is not predominant, spending energy in probing tends to hit the AoI performance and can sometimes render it worse than no probing.

\subsubsection{Markovian System}
We consider the same system parameters as in  Section~\ref{subsection:markovian-sensor-single-process}.  Here also, Fig.~\ref{fig5} (b) exhibits a similar trade-off  between the improvement offered by probing and the energy spent in it, as seen for the i.i.d. model. 

\vspace{-10pt}
\begin{figure}[h]
  \begin{center}
 \subfloat[]{\includegraphics[height=4cm,width=7cm]{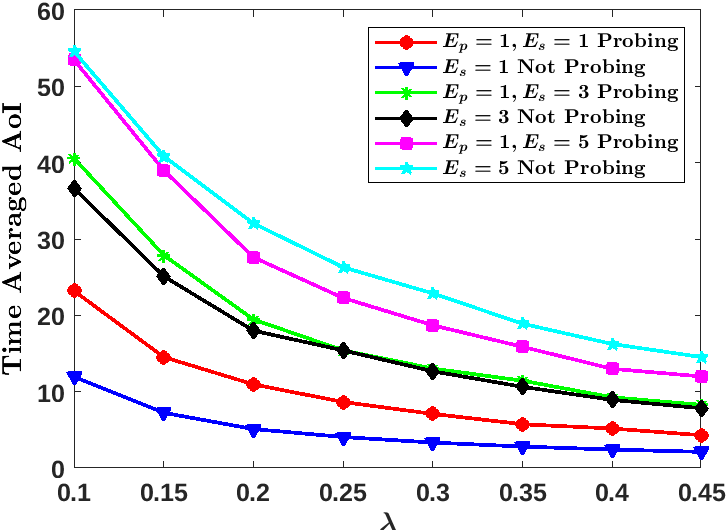}}
   \hfill
  \subfloat[]{\includegraphics[height=4cm,width=7cm]{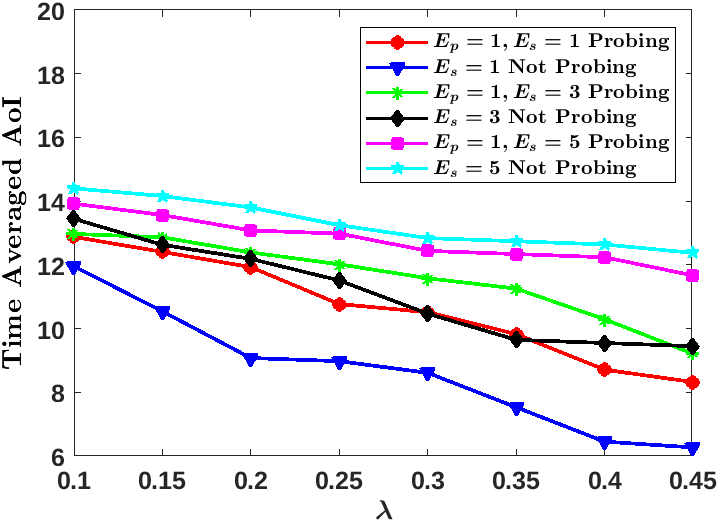}}
 \caption{For $N=1$, Variation of time-averaged AoI and comparison of a probing and non-probing system with the change in $\lambda, E_p$, and $E_s$ for (a) I.I.D. model,
 (b) Markovian model.
 }
 \label{fig5}
 \end{center}
 \vspace{-24pt}
\end{figure}

\subsection{Q-Learning}
We consider the same settings as Section~\ref{subsection:single-sensor-single=process-channel-fading} and Section~\ref{subsection:markovian-sensor-single-process}, except that we set $B=5$ and an upper bound on the age as  $T_{max} = 7$, in order to reduce the number of possible states. We then run an online Q-learning scheme across two different sample paths. For the sake of comparison, we also plot the time-averaged AoI of a uniform random strategy as well as that of the optimal policy derived through value iteration. Convergence of the time-averaged AoI under Q-learning to the optimal AoI for the i.i.d. channel system and the Markovian model can be seen in Fig. \ref{fig6} (a) and Fig. \ref{fig6} (b), respectively. We also observe significant improvement from the initial policy which always picks an action uniformly at random; this demonstrates the necessity of the MDP-based formulation for probing and sampling.

\vspace{-10pt}
\begin{figure}[h]
  \begin{center}
 \subfloat[]{\includegraphics[height=4cm,width=7cm]{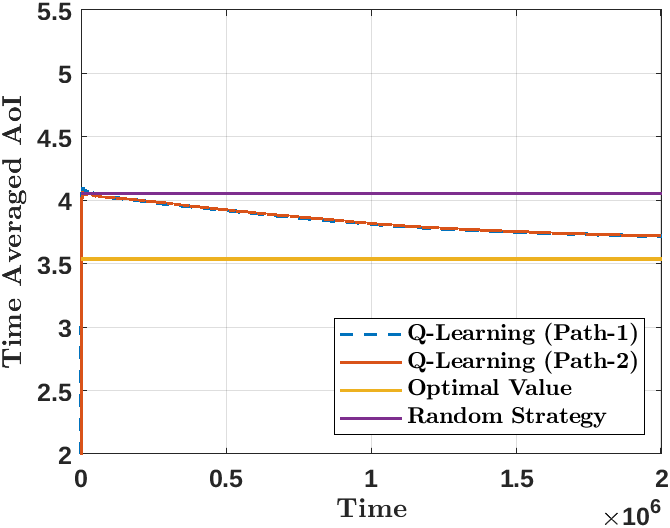}}
   \hfill
  \subfloat[]{\includegraphics[height=4cm,width=7cm]{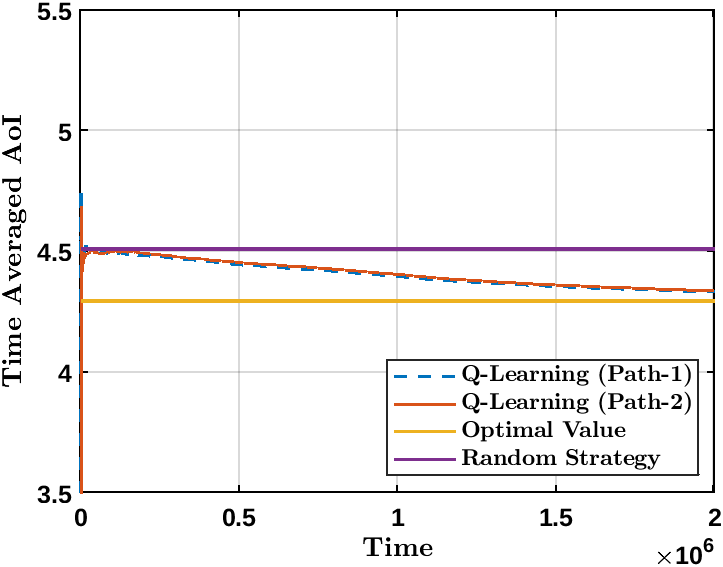}}
 \caption{Improvement in time-averaged AoI via   Q-learning, $N=1$, (a)    i.i.d. model, 
 (b)   Markovian model. 
 }
 \label{fig6}
 \end{center}
 \vspace{-12pt}
\end{figure}

\section{Conclusions}\label{section:conclusion}
In this paper, we have derived optimal policy structures for minimizing the time-averaged expected AoI under an energy-harvesting source. We considered single and multiple processes, i.i.d. and Markovian time-varying channels, time-varying characteristics for energy harvesting, and channel probing capability at the source. The optimal source sampling policy often turned out to be a threshold policy. We have also proposed RL algorithms for unknown channel statistics and energy harvesting characteristics. Numerical results demonstrated the power of channel probing, and also the performances and trade-offs for various algorithms.  However, there are a number of issues that remain untouched, such as multi-source star topology and multi-hop network settings with multiple source and sink nodes. We plan to address these issues in our future research.

\bibliographystyle{IEEEtran-mod}
\bibliography{ref_paper.bib}

\appendices
\section{Proof of proposition~\ref{proposition:convergence-of-value-function-J}}\label{appendix:proof-of-convergence-of-value-function-J}
We prove that $max_{s \in \mathcal{S}}|J^{(t)}(s)-J^{*}(s)| \uparrow 0$ as $t \uparrow \infty$. 

Let us define the error $e_{t}=\max_{s \in \mathcal{S}}|J^{(t)}(s)-J^{*}(s)|$ and $J^{(0)}(s)$ as initial estimate for $J^{*}(s)$.\\ 
For any state $s$, we seek to establish a relation between the error at time $t+1$ to the error at time $t$.

\scriptsize
\begin{eqnarray}\label{con1}
&&J^{(t+1)}(E \geq E_{p}+E_{s},T)\nonumber\\
&=& min \bigg\{T+\alpha \mathbb{E}_{A}J^{(t)}(min\{E+A,B\},T+1), \nonumber\\
&&\mathbb{E}_{C}\bigg(min\{ T+\alpha \mathbb{E}_{A}J^{(t)}(min\{E-E_{p}+A,B\},T+1), \nonumber\\
&&T(1-p(C))+\alpha p(C)\mathbb{E}_{A}J^{(t)}(min\{E-E_{p}-E_{s}+A,B\},1)+ \nonumber\\
&&\alpha (1-p(C))\mathbb{E}_{A}J^{(t)}(min\{E-E_{p}-E_{s}+A,B\},T+1)\} \bigg) \bigg\} 
\end{eqnarray}
\normalsize

We assume there exists an optimal value function $J^{*}(E,T)$ for the discounted cost problem. By using the fact that $J^{(t)}(E,T)$ by $J^{(t)}(E,T) \leq J^{(*)}(E,T) +e_{t}$ in \eqref{con1}, we obtain:

\vspace{-4pt}

\scriptsize
\begin{eqnarray}\label{con2}
&&J^{(t+1)}(E \geq E_{p}+E_{s},T)\nonumber\\ 
&\leq& min \bigg\{T+ \alpha \mathbb{E}_{A}(J^{*}(min\{E+A,B\},T+1)+e_{t}),\nonumber\\
&&\mathbb{E}_{C}\bigg(min\{ T+\alpha \mathbb{E}_{A}(J^{*}(min\{E-E_{p}+A,B\},T+1)+e_{t}), \nonumber\\
&&T(1-p(C))+\alpha p(C)\mathbb{E}_{A}(J^{*}(min\{E-E_{p}-E_{s}+A,B\},1)+e_{t})\nonumber\\
&&+\alpha (1-p(C))\mathbb{E}_{A}(J^{*}(min\{E-E_{p}-E_{s}+A,B\},T+1)+e_{t})\} \bigg) \bigg\}\nonumber\\
\end{eqnarray}
\normalsize

Similarly, using $J^{(t)}(E,T)$ by $J^{(t)}(E,T) \geq J^{(*)}(E,T) -e_{t}$ in \eqref{con1}, we obtain:

\vspace{-4pt}

\scriptsize
\begin{eqnarray}\label{con3}
&&J^{(t+1)}(E \geq E_{p}+E_{s},T)\nonumber\\
&\geq & min \bigg\{T+ \alpha \mathbb{E}_{A}(J^{*}(min\{E+A,B\},T+1)-e_{t}),\nonumber\\
&&\mathbb{E}_{C}\bigg(min\{ T+\alpha \mathbb{E}_{A}(J^{*}(min\{E-E_{p}+A,B\},T+1)-e_{t}), \nonumber\\
&&T(1-p(C))+\alpha p(C)\mathbb{E}_{A}(J^{*}(min\{E-E_{p}-E_{s}+A,B\},1)-e_{t})\nonumber\\
&&+\alpha (1-p(C))\mathbb{E}_{A}(J^{*}(min\{E-E_{p}-E_{s}+A,B\},T+1)-e_{t})\} \bigg) \bigg\}\nonumber\\
\end{eqnarray}
\normalsize

Combining the results obtained from equations \eqref{con2} and \eqref{con3}, we get:

\scriptsize
\begin{eqnarray}
&&J^{*}(E \geq E_{p}+E_{s},T)-\alpha e_{t} \nonumber\\
&\leq& J^{(t+1)}(E \geq E_{p}+E_{s},T) \leq J^{*}(E \geq E_{p}+E_{s},T)+\alpha e_{t} \nonumber\\
& \implies& |J^{(t+1)}(E \geq E_{p}+E_{s},T)-J^{*}(E \geq E_{p}+E_{s},T)| \leq \alpha e_{t} \nonumber\\
&\implies& \max_{s \in \mathcal{S}}|J^{(t+1)}(s)-J^{*}(s)| \leq \alpha e_{t} \nonumber\\
&\implies& e_{t+1} \leq \alpha e_{t} \nonumber\\
& \implies& \lim_{t \rightarrow \infty} e_t =0
\end{eqnarray}
\normalsize
 
  Hence, $J^{(t)}(s)$ converges to $J^{*}(s)$.

\vspace{-6pt}
\section{Proof of Lemma~\ref{lemma:single-sensor-single-process-J-decreasing-in-p}}\label{appendix:proof-of-lemma-single-sensor-single-process-fading-cost-increasing-in-T-decrease-in-p}
 We prove this result by value iteration:
 
\scriptsize
\begin{eqnarray}
&&J^{(t+1)}(E \geq E_{p}+E_{s},T)\nonumber\\
&=& min \bigg\{T+\alpha \mathbb{E}_{A}J^{(t)}(min\{E+A,B\},T+ 1), \nonumber\\
&&\mathbb{E}_{C}\bigg(min\{ T+\alpha\mathbb{E}_{A}J^{(t)}(min\{E-E_{p}+A,B\}, T+1),\nonumber\\
&&T(1-p(C))+\alpha p(C)\mathbb{E}_{A}J^{(t)}(min\{E-E_{p}-E_{s}+ A,B\},1)+\nonumber\\
&&\alpha (1-p(C))\mathbb{E}_{A}J^{(t)}(min\{E-E_{p}-E_{s}+A,B\},T+1)\} \bigg) \bigg\} \nonumber\\
&&J^{(t+1)}(E < E_{p}+E_{s},T)\nonumber\\
&=&T+\alpha \mathbb{E}_{A}J^{(t)}(min\{E+A,B\},T+1)
\end{eqnarray}
\normalsize

Let us start with $J^{(0)}(s) = 0$ for all $s\in S$. Clearly,
$J^{(1)}(E \geq E_{p}+E_{s}, T) =min \{T, \mathbb{E}_{C}(min\{ T,T(1-p(C))\}) \}$ $=min \{T, \mathbb{E}_{C}(T(1-p(C))\} $ and $J^{(1)}(E < E_{p}+E_{s} , T) = T$. Hence, for any given $E$, the value function $J^{(1)}(E, T)$ is an increasing function of $T$. As an induction hypothesis, we assume that $J^{(t)}(E, T)$ is also an increasing function of $T$. Now,

\scriptsize
\begin{eqnarray}\label{value-iteration-single-sensor-with-fading-1}
&&J^{(t+1)}(E \geq E_{p}+E_{s},T)\nonumber\\
&=& min \bigg\{T+\alpha \mathbb{E}_{A}J^{(t)}(min\{E+A,B\},T+1),\nonumber\\
&&\mathbb{E}_{C}\bigg(min\{ T+\alpha \mathbb{E}_{A}J^{(t)}(min\{E-E_{p}+A,B\},T+1),\nonumber\\
&& T(1-p(C))+\alpha p(C)\mathbb{E}_{A}J^{(t)}(min\{E-E_{p}-E_{s}+A,B\},1)+ \nonumber\\
&&\alpha (1-p(C))\mathbb{E}_{A}J^{(t)}(min\{E-E_{p}-E_{s}+A,B\},T+1)\} \bigg) \bigg\} 
\end{eqnarray}
\normalsize

We need to show that $J^{(t+1)}(E \geq E_{p}+E_{s},T)$ is also increasing in $T$. The first term inside the minimization operation in  \eqref{value-iteration-single-sensor-with-fading-1} is increasing in $T$, by the induction hypothesis and from the fact that expectation is a linear operation. On the other hand, the second term has linear expectation over the channel state and another minimization operator. Also, the first and second terms inside the second minimization operation in \eqref{value-iteration-single-sensor-with-fading-1} are increasing in $T$ by the induction hypothesis and the linearity of expectation operation. Thus, $J^{(t+1)}(E \geq E_{p}+E_{s},T)$ is also increasing in $T$. By similar arguments, we can claim that $J^{(t+1)}(E < E_{p}+E_{s} ,T)$ is increasing in $T$. Now, since $J^{(t)}(\cdot)\uparrow J^{*}(\cdot)$ as $t \uparrow \infty$ 
by proof of Proposition~\ref {proposition:convergence-of-value-function-J}, $J^{*}(E,T)$ is also increasing in $T$. Hence, the first part of the lemma is proved. 
\\
For current probed channel state, the value function $W^{(t+1)}(E,T,C)$ is given by: 

\vspace{-8pt}

\scriptsize
\begin{eqnarray}\label{value-iteration-single-sensor-with-fading-11}
&&W^{(t+1)}(E,T,C)\nonumber\\
&=&min\{T+\alpha \mathbb{E}_{A}J^{(t)}(min\{E-E_{p}+A,B\},T+1),\nonumber\\
&&T(1-p(C))+\alpha p(C)\mathbb{E}_{A}J^{(t)}(min\{E-E_{p}- E_{s}+A,B\},1)+\nonumber\\
&&\alpha (1-p(C))\mathbb{E}_{A}J^{(t)}(min\{E-E_{p}-E_{s}+A,B\},T+1)\}
\end{eqnarray}
\normalsize

The first term (cost of not sampling a source) inside the minimization operation in \eqref{value-iteration-single-sensor-with-fading-11} is independent of $p(C)$, whereas the second term (cost of sampling a source) inside the minimization operation in \eqref{value-iteration-single-sensor-with-fading-11} is given by:

\scriptsize
\begin{eqnarray}\label{value-iteration-single-sensor-with-fading-2}
&& T(1-p(C))+\alpha p(C)\mathbb{E}_{A}J^{(t)}(min\{E-E_{p}-E_{s}+A,B\},1)\nonumber\\
&&+ \alpha (1-p(C))\mathbb{E}_{A}J^{(t)}(min\{E-E_{p}-E_{s}+A,B\},T+1)\nonumber\\
&=& T(1-p(C))+\alpha \mathbb{E}_{A}J^{(t)}(min\{E-E_{p}-E_{s}+A,B\},T+1)\nonumber\\
&&-p(C)\alpha\bigg(\mathbb{E}_{A}J^{(t)}(min\{E-E_{p}-E_{s}+A,B\},T+1)-\nonumber\\
&&\mathbb{E}_{A}J^{(t)}(min\{E-E_{p}-E_{s}+A,B\},1)\bigg)
 \end{eqnarray}
\normalsize

By the induction hypothesis and the linearity of expectation operation, $\mathbb{E}_{A}J^{(t)}(min\{E-E_{p}-E_{s}+A,B\},T+1)- \mathbb{E}_{A}J^{(t)}(min\{E-E_{p}-E_{s}+A,B\},1)$ is non-negative. Thus, the second term inside the minimization operation in  \eqref{value-iteration-single-sensor-with-fading-11} is decreasing in $p(C)$. Now, since $W^{(t)}(\cdot)\uparrow W^{*}(\cdot)$ as $t \uparrow \infty$, $W^*(E,T,C)$ is also decreasing in $p(C)$.

\section{Proof of Theorem~\ref{theorem:single-sensor-single-process-with-fading-policy-structure}}\label{appendix:proof-of-theorem-single-sensor-single-process-fading-threshold-policy}

From \eqref{eqn:Bellman-eqn-single-sensor-single-process-with-fading}, it is obvious that for probed channel state the optimal decision for $E \geq E_{p}+E_{s} $ is to sample the source if and only if the cost of sampling is lower than the cost of not sampling the source, i.e., 
$T+\alpha \mathbb{E}_{A}J^{*}(min\{E-E_{p}+A,T+1) \geq T(1-p(C))+\alpha \mathbb{E}_{A}J^{*}(min\{E-E_{p}-E_{s}+A,B\},T+1)-
\alpha p(C)\bigg(\mathbb{E}_{A}J^{*}(min\{E-E_{p}-E_{s}+A,B\},T+1)-\mathbb{E}_{A}J^{*}(min\{E-E_{p}-E_{s}+A,B\},1)\bigg)$.
Now, by Lemma~\ref{lemma:single-sensor-single-process-J-decreasing-in-p}, $\mathbb{E}_{A}J^{*}(min\{E-E_{p}-E_{s}+A,B\},T+1)- \mathbb{E}_{A}J^{*}(min\{E-E_{p}-E_{s}+A,B\},1)$ is non-negative. Thus the R.H.S. decreases with $p(C)$, whereas the L.H.S. is independent of $p(C)$. Hence, for probed channel state the optimal action is to sample if and only if $p(C) \geq p_{th}(E,T)$ for some suitable threshold function $p_{th}(E,T)$.

\section{Proof of Lemma~\ref{lemma:single-sensor-multiple-process-J-increasing-in-T-decreasing-in-p(C^')}}\label{appendix:proof-of-lemma-single-sensor-multiple-process-fading-cost-increasing-in-T-decrease-in-p}
The proof is similar to the proof of Lemma~\ref{lemma:single-sensor-single-process-J-decreasing-in-p} and it follows from the convergence of value iteration as given below:

\scriptsize
\begin{eqnarray}\label{value-iteration-single-sensor-multiple-process-with-fading}
&&J^{(t+1)}(E \geq E_{p}+E_{s},T_1, T_2,\cdots, T_N)\nonumber\\
&=&min \bigg\{\sum_{i=1}^N T_i+\alpha \mathbb{E}_{A}J^{(t)}(min\{E+A,B\},T_1+1, T_2+1, \cdots,\nonumber\\
&&T_N+1),\mathbb{E}_{C} \bigg( min\{\sum_{i=1}^N T_i+\alpha \mathbb{E}_{A}J^{(t)}(min\{E-E_{p}+A,B\}, \nonumber\\ 
&&T_1+1,T_2+1,  \cdots,T_N+1),\min_{1 \leq k \leq N} \bigg( \sum_{i \neq k } T_i+T_k(1-p(C))+\nonumber\\
&&\alpha p(C)\mathbb{E}_{A}J^{(t)}(min\{E-E_{p}-E_{s}+A,B\},T_1+1,T_2+1, \cdots,  \nonumber\\
&&T_k'=1,T_{k+1}+1,\cdots,T_N+1)+\alpha (1-p(C))\mathbb{E}_{A}J^{(t)}(min\{E-\nonumber\\
&&E_{p}-E_{s}+A,B\},T_1+1,T_2+1, \cdots,T_N+1) \bigg)\} \bigg) \bigg\} \nonumber\\
&&J^{(t+1)}(E < E_{p}+E_{s},T_1, T_2,\cdots, T_N)\nonumber\\
&=&\sum_{i=1}^N T_i+\alpha \mathbb{E}_{A}J^{(t)}(min\{E+A,B\},T_1+1, T_2+1, \cdots,T_N+1)\nonumber\\ 
\end{eqnarray}
\normalsize

Let us start with $J^{(0)}(s) = 0$ for all $s\in S$. Clearly,
$J^{(1)}(E \geq E_{p}+E_{s}, T_1, T_2,\cdots, T_N) = \min \{\sum_{i=1}^N T_i, \mathbb{E}_{C}(min\{\sum_{i=1}^N T_i,\min_ {1 \leq k \leq N} ( \sum_{i \neq k } T_i+T_k(1-p(C)))\}\}$ and $J^{(1)}(E < E_{p}+E_{s}, T_1, T_2,\cdots, T_N) = \sum_{i=1}^N T_i$. Hence, for any given $E$, the value function $J^{(1)}(E,T_1, T_2,\cdots, T_N)$ is an increasing function of $T_1, T_2, \cdots, T_N $. As induction hypothesis, we assume that $J^{(t)}(E,T_1, T_2,\cdots, T_N)$ is also increasing function of $T_1, T_2, \cdots, T_N $. 

Now,

\scriptsize
\begin{eqnarray}\label{value-iteration-single-sensor-multiple-process-with-fading-1}
&&J^{(t+1)}(E \geq E_{p}+E_{s},T_1, T_2,\cdots, T_N)\nonumber\\
&=&min \bigg\{\sum_{i=1}^N T_i+\alpha \mathbb{E}_{A}J^{(t)}(min\{E+A,B\},T_1+1, T_2+1, \cdots,\nonumber\\
&&T_N+1),\mathbb{E}_{C} \bigg( min\{\sum_{i=1}^N T_i+\alpha \mathbb{E}_{A}J^{(t)}(min\{E-E_{p}+A,B\},  \nonumber\\ 
&&T_1+1,T_2+1,\cdots, T_N+1),\min_{1 \leq k \leq N} \bigg( \sum_{i \neq k } T_i+T_k(1-p(C))+\nonumber\\
&&\alpha p(C)\mathbb{E}_{A}J^{(t)}(min\{E-E_{p}-E_{s}+A,B\},T_1+1,T_2+1, \cdots,  \nonumber\\
&&T_k'=1,T_{k+1}+1,\cdots,T_N+1)+\alpha (1-p(C))\mathbb{E}_{A}J^{(t)}(min\{E-\nonumber\\
&&E_{p}-E_{s}+A,B\},T_1+1,T_2+1, \cdots,T_N+1) \bigg)\} \bigg) \bigg\} 
\end{eqnarray}
\normalsize

We seek to show that $J^{(t+1)}(E \geq E_{p}+E_{s} ,T_1, T_2,\cdots, T_N)$ is also increasing in each of $T_1, T_2, \cdots, T_N$. The first term inner to the minimization operation in \eqref{value-iteration-single-sensor-multiple-process-with-fading-1} is increasing in each of $T_1, T_2, \cdots, T_N$, utilizing the induction hypothesis and linear property of expectation operation. On the other hand, the second term has expectation over the channel state and another minimization operator. Also,  first and second terms of the second minimization operator are increasing in each of $T_1, T_2, \cdots, T_N$ by using the induction hypothesis and the linearity of expectation operation. Thus, $J^{(t+1)}(E\geq E_{p}+E_{s} , T_1, T_2,\cdots, T_N)$ is also increasing in each of $T_1, T_2, \cdots, T_N$. 
Similarly, we can assert that $J^{(t+1)}(E < E_{p}+E_{s},T_1, T_2,\cdots, T_N)$ is increasing in each of $T_1, T_2, \cdots, T_N$.
Now, since $J^{(t)}(\cdot)\uparrow J^{*}(\cdot)$ as $t \uparrow \infty$, $J^*(E,T_1, T_2,\cdots, T_N)$ is also increasing in each of $T_1, T_2,\cdots, T_N$. Hence, the first part of the lemma is proved.
Proof of second part of Lemma~\ref{lemma:single-sensor-multiple-process-J-increasing-in-T-decreasing-in-p(C^')} is similar to the proof of second part of Lemma~\ref{lemma:single-sensor-single-process-J-decreasing-in-p}

\section{Proof of Theorem ~\ref{theorem:single-sensor-multiple-process-with-fading-policy-structure}}\label{appendix:proof-of-theorem-single-sensor-multiple-process-fading-threshold-policy}
It is obvious that $J^*(\cdot)$ is invariant to any permutation of $(T_1,T_2,\cdots, T_N)$. Hence, by Lemma~\ref{lemma:single-sensor-multiple-process-J-increasing-in-T-decreasing-in-p(C^')}, $ \arg \min_{1 \leq k \leq N}\bigg(\sum_{i \neq k } T_i+T_k(1-p(C))+\alpha p(C)\mathbb{E}_{A}J^{(t)}(min\{E-E_{p}-E_{s}+A,B\},T_1+1,T_2+1, \cdots, T_k'=1,T_{k+1}+1,\cdots,T_N+1)+\alpha (1-p(C))\mathbb{E}_{A}J^{(t)}(min\{E-E_{p}-E_{s}+A,B\},T_1+1,T_2+1, \cdots,T_N+1) \bigg)=\arg \max_{1 \leq k \leq N} T_k$, i.e., the best process to activate is $k^* \doteq \arg \max_{1 \leq k \leq N} T_k$ in case one process has to be activated. For probed channel state, it is optimal to sample a process if and only if the cost of sampling this process is less than or equal to the cost of not sampling this process, which translates into $T_{k^*}+\alpha \mathbb{E}_{A}J^{*}(min\{E-E_{p}+A,B\},T_1+1, T_2+1,\cdots, T_N+1) \geq T_{k^*}(1-p(C))+
\alpha \mathbb{E}_{A}J^{*}(min\{E-E_{p}-E_{s}+A,B\},T_1+1,T_2+1, \cdots,T_N+1)-\alpha p(C)\bigg(\mathbb{E}_{A}J^{*}(min\{E-E_{p}-E_{s}+A,B\},T_1+1,T_2+1, \cdots,T_N+1)-\mathbb{E}_{A}J^{*}(min\{E-E_{p}-E_{s}+A,B\},T_1+1,T_2+1, \cdots, T_k'=1,T_{k+1}+1,\cdots,T_N+1)\bigg)$. Now, by Lemma~\ref{lemma:single-sensor-multiple-process-J-increasing-in-T-decreasing-in-p(C^')}, $\mathbb{E}_{A}J^{*}(min\{E-E_{p}-E_{s}+A,B\},T_1+1,T_2+1, \cdots,T_N+1)-\mathbb{E}_{A}J^{*}(min\{E-E_{p}-E_{s}+A,B\},T_1+1,T_2+1, \cdots, T_k'=1,T_{k+1}+1,\cdots,T_N+1) $ is non negative. Thus, the R.H.S. is decreasing in $p(C)$ and the L.H.S. is independent of $p(C)$. Hence, the threshold structure of the optimal sampling policy is proved.
\normalsize


\end{document}